\documentclass{vldb}
\usepackage{balance}  

\usepackage{cite} 
\usepackage{amssymb}
\usepackage{amsmath}
\usepackage{nicefrac} 
\usepackage{mathtools} 
\usepackage{enumerate}
\usepackage{latexsym}
\usepackage{float}
\usepackage{subcaption} 
\usepackage{color}
\usepackage{xspace}
\usepackage{url}
\usepackage{paralist} 
\usepackage{hyperref}
\usepackage{mathdots}
\usepackage{multirow} 
\usepackage{balance}

\usepackage{tikz}
\usepackage{listings}
\usepackage{pdfpages}

\usepackage{algorithm}
\usepackage[noend]{algpseudocode}

\newcommand{\citefullpaperA}{Appendix~\ref{appA}}
\newcommand{\citefullpaperB}{Appendix~\ref{appB}}
\newcommand{\ourend}{\ensuremath{\Box}}

\newcommand{\specialitem}{\ensuremath{\circ}}

\newdef{definition}{Definition}
\newdef{example}{Example}
\newdef{remark}{Remark}
\newdef{note}{Note}

\newtheorem{theorem}{Theorem}

\newcommand{\ignore}[1]{}

\newcommand{\bsgfopt}{{\sc BSGF-Opt}\xspace}
\newcommand{\sgfopt}{{\sc SGF-Opt}\xspace}

\newcommand{\sgfgopt}{{\sc Subset Cost-Opt}\xspace}
\newcommand{\sgfg}{{\sc Subset Cost}\xspace}
\newcommand{\sgf}{{\sc SGF}\xspace}

\newcommand{\db}{\textsf{\upshape DB}}

\newcommand{\bodyquery}{\xi}

\newcounter{myenum}

\newcommand{\opt}{\textsf{OPT}\xspace}
\newcommand{\gopt}{\textsf{GOPT}\xspace}

\newcommand{\map}{\text{\it map}}
\newcommand{\red}{\text{\it red}}

\newcommand{\mystart}{\text{\it h}}

\newcommand{\mr}{MR\xspace}

\newcommand{\out}{\text{\it out}}

\newcommand{\conforms}{\models}

\newcommand{\camera}[1]{#1}

\newcommand{\cF}{\mathcal{F}}
\newcommand{\cG}{\mathcal{G}}

\newcommand{\cI}{\mathcal{I}}

\newcommand{\cS}{\mathcal{S}}

\newcommand{\cX}{\mathcal{X}}

\newcommand{\gsgf}{{\sc Greedy-SGF}\xspace}
\newcommand{\gbsgf}{{\sc Greedy-BSGF}\xspace}

\newcommand{\va}{\bar a}
\newcommand{\vb}{\bar{b}}
\newcommand{\vc}{\bar{c}}

\newcommand{\vu}{\bar{u}}
\newcommand{\vv}{\bar{v}}

\newcommand{\vt}{\bar{t}}
\renewcommand{\vv}{\bar{v}}
\newcommand{\vw}{\bar{w}}
\newcommand{\vx}{\bar{x}}
\newcommand{\vy}{\bar{y}}
\newcommand{\vz}{\bar{z}}

\DeclareMathOperator{\overlap}{overlap}

\newcommand{\fD}{{\bf D}}
\newcommand{\allvars}{{\bf V}}
\newcommand{\schema}{{\bf S}}

\newcommand{\guarded}{conditional\xspace}

\newcommand{\validate}{\textsf{MSJ}\xspace}

\newcommand{\validatenospace}{\textsf{MSJ}}

\newcommand{\eval}{\textsf{EVAL}\xspace}

\newcommand{\Gumbo}{Gumbo\xspace}

\newcommand{\SEQ}{\textsc{SEQ}\xspace}
\newcommand{\PAR}{\textsc{PAR}\xspace}
\newcommand{\ONEROUND}{\textsc{1-ROUND}\xspace}
\newcommand{\MSJ}{\textsc{GREEDY}\xspace}
\newcommand{\HIVEPAR}{\textsc{HPAR}\xspace}

\newcommand{\HIVEPARSEM}{\textsc{HPARS}\xspace}
\newcommand{\PIGPAR}{\textsc{PPAR}\xspace}

\newcommand{\GCOST}{\MSJ}

\newcommand{\PARUNIT}{\textsc{PARUNIT}\xspace}
\newcommand{\SEQUNIT}{\textsc{SEQUNIT}\xspace}

\newcommand{\semijoin}{\ensuremath{\ltimes}}
\newcommand{\antijoin}{\ensuremath{\rhd}}
\newcommand{\semi}[3]{\ensuremath{\pi_{#1}( #2 \semijoin #3) }}

\usepackage{textcomp}

\newcommand{\OMIT}[1]{}

\DeclarePairedDelimiter{\ceil}{\lceil}{\rceil}

\newcommand{\jonnystilde}{{\raise.17ex\hbox{$\scriptstyle\sim$}}}

\newcommand{\kv}[2]{\ensuremath{\left\langle#1:#2\right\rangle }}
\newcommand{\kvpair}[2]{\ensuremath{\left\langle#1:#2\right\rangle }}

\newcommand{\ASSERT}{\textsc{Assert}\xspace}
\newcommand{\REQUEST}{\textsc{Request}\xspace}

\newcommand{\CONFIRM}{\textsc{Confirm}\xspace}
\newcommand{\DENY}{\textsc{Deny}\xspace}

\newcommand{\REQ}{\textsc{Req}}
\newcommand{\HYP}{\textsc{Out}}
\newcommand{\OUT}{\textsc{Out}}
\newcommand{\RHLIST}{\textsc{Req-List}}

\newcommand{\ASLIST}{\textsc{Assert-List}}

\newcommand{\Outrel}{Z}
\newcommand{\outrel}{\Outrel}

\DeclareMathOperator{\pack}{pack}

\newcommand{\msj}{\mathop{\ltimes\!\cdot}}

\newcommand{\inp}{\text{\it inp}}
\newcommand{\inter}{\text{\it int}}

\newcommand{\Lrc}{\ensuremath{l_r}}
\newcommand{\Lwc}{\ensuremath{l_w}}
\newcommand{\Hrc}{\ensuremath{h_r}}
\newcommand{\Hwc}{\ensuremath{h_w}}
\newcommand{\Trc}{\ensuremath{t}}

\newcommand{\MoutM}{\ensuremath{\widehat{M}_i}}
\newcommand{\mapn}{\ensuremath{m_i}}
\newcommand{\redn}{\ensuremath{r}}
\newcommand{\mbuflim}{\ensuremath{\mathit{buf}_{\map}}}

\newcommand{\rbuflim}{\ensuremath{\mathit{buf}_{\red}}}

\DeclareMathOperator{\mapmerge}{merge_{\map}}
\DeclareMathOperator{\redmerge}{merge_{\red}}

\DeclareMathOperator{\cost}{\text{\it cost}}
\DeclareMathOperator{\gain}{gain}

\newcommand{\ourcost}{\ensuremath{\cost_{\mathit{gumbo}}}\xspace}
\newcommand{\wangscost}{\ensuremath{\cost_{\mathit{wang}}}\xspace}

\newcommand{\SELECT}{{\tt SELECT}}
\newcommand{\FROM}{{\tt FROM}}

\newcommand{\WHERE}{{\tt WHERE}}
\newcommand{\NOT}{{\tt NOT}}
\newcommand{\AND}{{\tt AND}}
\newcommand{\OR}{{\tt OR}}

\floatstyle{boxed}
\newfloat{algo}{thp}{lop}
\floatname{algo}{Algorithm}

\begin{document}


 \title{Parallel Evaluation of Multi-Semi-Joins}
\numberofauthors{2} 
\author{
\alignauthor Jonny Daenen\\
       \affaddr{Hasselt University}\\
       \email{jonny.daenen@uhasselt.be}       
\alignauthor Frank Neven\\
       \affaddr{Hasselt University}\\
       \email{frank.neven@uhasselt.be}
\and
\alignauthor  Tony Tan\\
       \affaddr{National Taiwan University}\\
       \email{tonytan@csie.ntu.edu.tw}
\alignauthor  Stijn Vansummeren\\
       \affaddr{Universit\'e Libre de Bruxelles}\\
       \email{stijn.vansummeren@ulb.ac.be}
}
\date{}

\maketitle


\begin{abstract}
{While services such as Amazon AWS make
computing power abundantly available, adding more computing nodes 
can incur high costs in, for instance, pay-as-you-go plans while
not always significantly improving the net running time (aka wall-clock time) of queries.} 
In this work, 
we provide algorithms for
parallel evaluation of SGF queries in
MapReduce that optimize total time, while retaining low net time. 
Not only can SGF queries specify all semi-join reducers, but also more expressive
queries involving disjunction and negation. Since SGF queries can be
seen as Boolean combinations of (potentially nested) semi-joins, we
introduce a novel multi-semi-join (MSJ) MapReduce operator that
enables the evaluation of a set of semi-joins in one job. 
We use \camera{this} operator to obtain parallel query plans for SGF queries that
outvalue sequential plans w.r.t.\ net time and provide additional
optimizations aimed at minimizing total time without severely affecting net
time. Even though the latter optimizations are NP-hard, we present
effective greedy algorithms. {Our experiments, conducted using our
own implementation Gumbo on top of Hadoop, confirm the 
usefulness of parallel query plans, and the effectiveness and scalability of
our optimizations, all with a significant improvement over Pig and Hive.}
\end{abstract}

\section{Introduction}
\label{sec:intro}

The problem of evaluating joins efficiently in massively parallel
systems is an active area of
research (e.g., \cite{Elseidy:2014ui,Okcan:2011kb,Chu:2015de,AfratiSSU13,AfratiU11,KoutrisS11,BKS13,BKS14,AfratiUV15,AfratiJRSU-arxiv14}).
Here, efficiency can be measured in terms of
different criteria, including net time, total time, amount of
communication, resource requirements and the number of synchronization
steps. 
\camera{As parallel systems aim to bring down the net time, 
i.e., the difference between query end and start time,
it is often considered the most important criterium.}
The amount of computing power is no longer an issue through 
the readily availability of services such as Amazon AWS. 
However, in pay-as-you-go plans,
the cost is determined by the total time, that is, 
the aggregate sum of time spent by \emph{all} computing nodes. 
In this paper, we focus on parallel evaluation of queries that 
minimize total time while retaining low net time.
\camera{We consider parallel query plans that exhibit low net times
and exploit commonalities between queries to bring down the total time.}

Semi-joins have played a fundamental role in minimizing
communication costs in traditional database systems through their role
in semi-join reducers~\cite{BernsteinC81,BernsteinG81},
\camera{facilitating the reduction of communication in multi-way join computations.}
In more recent
work, Afrati et al.~\cite{AfratiJRSU-arxiv14} provide an algorithm for
computing $n$-ary joins in MapReduce-style systems in which semi-join
reducers play a central role. Motivated by the general importance of semi-joins, we study the system aspects of implementing semi-joins in a MapReduce context.  In particular, we introduce a multi-semi-join operator \validate that enables the evaluation of a set of semi-joins in one Mapreduce job while reducing resource usage like total time and requirements on cluster
size without sacrificing net
time. We then use this operator to efficiently evaluate Strictly Guarded Fragment (SGF) queries~\cite{FlumFG02,ABN98}. Not only can this query language
specify all semi-join reducers, but also more expressive
queries involving disjunction and negation.


We illustrate our approach by means of a simple example. Consider the following SGF query $Q$:
$$
\begin{array}{l}
\SELECT\ (x,y)\ \FROM\ R(x,y)\\
\WHERE\ \bigl(S(x,y)\ \OR\ S(y,x)\bigr)\ \AND\ T(x,z)
\end{array}
$$
Intuitively, this query asks for all pairs $(x,y)$ in $R$ for which
there exists some $z$ such that (1) $(x,y)$ or $(y,x)$ occurs in $S$
and (2) $(x,z)$ occurs in $T$. To evaluate $Q$ it suffices to compute
the following semi-joins
$$
\begin{array}{l}
X_1 \quad \coloneqq \quad R(x,y)\semijoin S(x,y);\\
X_2 \quad \coloneqq \quad R(x,y)\semijoin S(y,x); \\ 
X_3 \quad \coloneqq \quad R(x,y)\semijoin T(x,z);
\end{array}
$$
store the results in the binary relations $X_1$,
$X_2$, or $X_3$, and subsequently compute $\varphi := (X_1 \cup X_2)
\cap X_3$.
Our multi-semi-join operator $\validate({\cal S})$ (defined in
Section~\ref{algo:validate}) takes a number of semi-join-equations as
input and exploits commonalities between them to optimize
evaluation. In our framework, a possible query plan for query $Q$ is of the form:
\begin{center}
\includegraphics{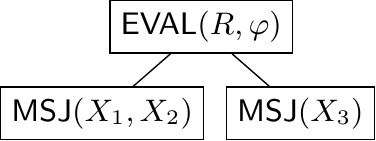}
\end{center}
In this plan, the
calculation of $X_1$ and $X_2$ is combined in a single MapReduce job;
$X_3$ is calculated in a separate job; and \camera{$\eval(R, \varphi)$} is a third
job responsible for computing \camera{the subset of $R$ defined by} $\varphi$.
We provide a cost model to determine the best query plan for SGF
queries. We note that, unlike the simple query $Q$ illustrated here,
SGF queries can be nested in general. In addition, we also show how to generalize the method to the simultaneous evaluation of multiple SGF queries.


    

The contributions of this paper can be summarized as follows:
\begin{compactenum}
\item We introduce the multi-semi-join operator $\msj(\cS)$ to evaluate a set $\cS$ of semi-joins and present a corresponding MapReduce implementation $\validate(\cS)$.

\item We present query plans for basic, that is, unnested, SGF queries and 
{propose an improved version of the cost model presented by}~\cite{wang2013,nykiel2010} for estimating their cost. As computing the optimal plan for a given basic SGF query is NP-hard, we \camera{provide} a fast greedy heuristic \gbsgf.

\item We show that the evaluation of (possibly nested) SGF queries can be reduced to the evaluation of a set of basic SGF queries in an order consistent with the dependencies induced by the former. In this way, computing an optimal plan for a given SGF query (which is NP-hard as well) can be \camera{approximated} \camera{by first determining an optimal subquery ordering, followed by an optimal evaluation of these subqueries. For the former, we present a greedy algorithm called \gsgf.}

\item 

We experimentally assess the effectiveness of \gbsgf and \gsgf
and obtain that, backed by an updated cost model, 
these algorithms successfully manage to bring down 
total times of parallel evaluation, \camera{making it comparable to that
of sequential query plans, while still retaining low net times. 
This is especially} true in the 
presence of commonalities among the atoms of 
queries. \camera{Finally, our system outperforms} Pig and Hive in all aspects 
when it comes to parallel evaluation of SGF queries \camera{and
displays interesting scaling characteristics.}
\end{compactenum}

\medskip

\noindent
{\bf Outline.}
This paper is organized as follows.
We discuss related work in Section~\ref{sec:related-works}.
We introduce the strictly guarded fragment (SGF) queries and discuss MapReduce 
{and the accompanying cost model}
in Section~\ref{sec:prelim}. In Section~\ref{sec:eval}, we consider the evaluation of multi-semi-joins and SGF queries. In Section~\ref{sec:experiments}, we discuss the experimental validation. We conclude in Section~\ref{sec:concl}.

\section{Related work}
\label{sec:related-works}

Recall that first-order logic (FO) queries are equivalent in expressive power
to the relational algebra (RA) and form the core fragment of SQL queries (cf., e.g., \cite{AHV95}).
Guarded-fragment (GF) queries have been studied extensively by the logicians in 1990s and 2000s,
and they emerged from the intensive efforts to obtain a syntactical classification
of FO queries with decidable satisfiability {problems}.
For more details, we refer the reader to the highly influential paper by 
Andreka, van Benthem, and Nemeti~\cite{ABN98}, as well as
survey papers by Gr\"adel and Vardi~\cite{Gradel98,Vardi96}.
In traditional database terms,
GF queries are equivalent in expressive power to semi-join algebra~\cite{LMTV05}.
Closely related are freely acyclic GF queries, which are GF queries
restricted to using only the $\wedge$ operator and guarded existential
quantifiers~\cite{PicalausaFHV14}. 
Flum et al.{}\cite{FlumFG02} introduced the term 
strictly guarded fragment queries for queries of the form $\exists \vy(\alpha\land\varphi)$. That is, guarded fragment queries without Boolean combinations at the outer level. We consider a slight generalization of these queries as explained in Remark~\ref{rem:strict:extension}. 

In general, obtaining the optimal plan in SQL-like query evaluation,
even in centralized computation, is a hard problem~\cite{Chaudhuri98,Ioannidis96}. Classic works by Yannakakis and Bernstein advocate the use of semi-join operations to optimize the evaluation of conjunctive queries~\cite{Yannakakis81,BernsteinC81,BernsteinG81}. A lot of work has been invested to optimize query evaluation in Pig~\cite{GatesDN13,OlstonRSS08}, Hive~\cite{hive} and SparkSQL~\cite{XinRZFSS13} as well as in MapReduce setting in general~\cite{Blanas:2010bj}.
None of them target SGF queries directly.

Tao, Lin and Xiao~\cite{TLX13} studied {\em minimal} MapReduce al\-go\-rithms,
i.e.\ al\-go\-rithms that scale linearly to the number of servers 
in all significant aspects of parallel computation
such as reduce compute time, bits of information received and sent, 
as well as storage space required by each server.
They show that, among many other problems, 
a single semi-join query between two relations can be evaluated by
a one round minimal algorithm.
This is a simpler problem, as a single basic SGF query may involve multiple semi-join queries. 
Afrati et al.\cite{AfratiJRSU-arxiv14} introduced a generalization of Yannakakis' al\-go\-rithm (using semi-joins) to a MapReduce setting.
Note that Yannakakis' algorithm starts with a sequence of semijoin operations, which
is a (nested) SGF query in a very restricted form.

\section{Preliminaries}
\label{sec:prelim}

We start by introducing the necessary concepts and terminologies.
In Section~\ref{sec:prelim:gf}, we define the strictly guarded fragment queries, while we discuss MapReduce in Section~\ref{subsec:MR}
and the cost model in Section~\ref{subsec:cost-model}.


\subsection{Strictly Guarded Fragment Queries}
\label{sec:prelim:gf}

In this section, we define the strictly guarded fragment queries (SGF)
\cite{FlumFG02}, but use a non-standard, SQL-like notation for ease of
readability. 

 We assume given a fixed infinite set $\fD =
\{a,b,\dots\}$ of data values and a fixed collection of relation
symbols $\schema = \{R,S,\dots\}$, disjoint with $\fD$. Every relation
symbol $R \in \schema$ is associated with a natural number called the
\emph{arity of $R$}. An expression of the form $R(\va)$ with $R$ a
relation symbol of arity $n$ and $\va \in \fD^n$ is called a
\emph{fact}.  A database $\db$ is then a finite set of facts. Hence, we
write $R(\va) \in \db$ to denote that a tuple $\va$ belongs to the $R$
relation in
$\db$.




We also assume given a fixed infinite set $\allvars = \{x,y,\dots\}$
of variables, disjoint with $\fD$ and $\schema$.  A \emph{term} is
either a data value or a variable. An \emph{atom} is an expression of
the form $R(t_1,\dots,t_n)$ with $R$ a relation symbol of arity $n$
and each of the $t_i$ a term, $i \in [1,n]$.  (Note that every fact is
also an atom.)  A \emph{basic strictly guarded fragment (BSGF) query}
(or just a basic query for short) is an expression of the form
\begin{equation}
  \label{eq:gf}
  \Outrel \coloneqq \SELECT \ \vx\ \FROM \ R(\vt)\ [ \ \WHERE \ C \ ];
\end{equation}
where $\vx$ is a sequence of variables that all occur in the atom
$R(\vt)$, and the $\WHERE\ C$ clause is optional. If it occurs, $C$
must be a Boolean combination of atoms. Furthermore, to ensure that
queries belong to the guarded fragment, we require that for each pair
of distinct atoms $S(\vu)$ and $T(\vv)$ in $C$ it must hold that all
variables in $\vu \cap \vv$ also occur in $\vt$.  (See also
Remark~\ref{rem:strict:extension} below.) The atom $R(\vt)$ is called the
\emph{guard} of the query, while the atoms occurring in $C$ are called
the \emph{\guarded atoms}. We interpret $Z$ as the output relation of
the query.



On a database $\db$, the BSGF query (\ref{eq:gf}) defines a new
relation $\Outrel$ containing all tuples $\va$ for which there is a
substitution $\sigma$ for the variables occurring in $\vt$ such that
$\sigma(\vx)=\va$, $R(\sigma(\vt)) \in \db$, and $C$ evaluates to true
in $\db$ under substitution $\sigma$. Here, the evaluation of $C$ in
$\db$ under $\sigma$ is defined by recursion on the structure of
$C$. If $C$ is $C_1 \text{\texttt{ OR }} C_2$, $C_1 \text{\texttt{ AND
  }} C_2$, or $\text{\texttt{NOT }} C_1$, the semantics is the usual
boolean interpretation. If $C$ is an atom $T(\vv)$ then $C$ evaluates
to true if $\sigma(\vt) \in R(\vt)\semijoin T(\vv)$, i.e., if there
exists a $T$-atom in $\db$ that equals $R(\sigma(\vt))$ on those
positions where $R(\vt)$ and $T(\vv)$ share variables. 

\begin{example}
The intersection $\Outrel_1 \coloneqq R \cap S$ and the difference $\Outrel_2 \coloneqq R - S$ 
between two relations $R$ and $S$ are expressed as follows:
\begin{eqnarray*}
\Outrel_1 & \coloneqq  & \SELECT \ \vx
\ \FROM \ R(\vx)
\ \WHERE \ S(\vx);
\\
\Outrel_2 & \coloneqq  & \SELECT \ \vx
\ \FROM \ R(\vx)
\ \WHERE \ \NOT \ S(\vx);
\end{eqnarray*}
The semijoin $\Outrel_3 = R(\vx,\vy)\semijoin S(\vy,\vz)$
and the antijoin $\Outrel_4 = R(\vx,\vy)\antijoin S(\vy,\vz)$ are expressed as follows:
\begin{eqnarray*}
\Outrel_3 & \coloneqq  & \SELECT \ \vx,\vy
\ \FROM \ R(\vx,\vy)
\ \WHERE \ S(\vy,\vz);
\\
\Outrel_4 & \coloneqq  & \SELECT \ \vx,\vy
\ \FROM \ R(\vx,\vy)
\ \WHERE \ \NOT \ S(\vy,\vz);
\end{eqnarray*}
The following BSGF query selects all the
pairs $(x,y)$ for which $(x,y,4)$ occurs in  $R$ and either $(1,x)$ or
$(y,10)$ is in $S$, but not both:
\begin{eqnarray*}
\Outrel_5 & \coloneqq  & \SELECT \ (x,y)
\ \FROM \ R(x,y,4)
\\
& & \WHERE \ (S(1,x) \ \AND \ \NOT \ S(y,10))
\\ 
& & \OR \ (\NOT \ S(1,x) \ \AND \ S(y,10));
\end{eqnarray*}
Finally, the traditional star semi-join between $R(x_1,\dots,x_n)$ and
relations $S_i(x_i,y_i)$, for $i\in[1,n]$, is expressed as
follows: 
\begin{eqnarray*}
\Outrel_6 & \coloneqq & \SELECT \ (x_1,\dots,x_n) \ \FROM \ R(x_1,\ldots,x_n)\\
  & & \WHERE \ S(x_1,y_1) \ \AND \dots \ \AND\ S(x_n,y_n); \qquad\hfill \ourend
\end{eqnarray*}
\end{example}



A \emph{strictly guarded fragment (SGF) query} is a collection
of BSGFs of the form
\camera{$
\Outrel_1  \coloneqq  \bodyquery_1; \ldots; \Outrel_n  \coloneqq \bodyquery_n;
$}
where each $\bodyquery_i$ is a BSGF that can mention any of the
predicates $\Outrel_j$ with $j<i$.  On a database $\db$, the SGF query
then defines a new relation $\Outrel_n$ where every occurrence of
$\Outrel_i$ is defined by evaluating $\bodyquery_i$.  

\begin{example} Let {\bf Amaz}, {\bf BN}, 
and {\bf BD} be
relations containing tuples (title, author, rating) corresponding to
the books found at Amazon, Barnes and Noble, and Book Depository,
respectively. Let \textbf{Upcoming} contain tuples (new\-title, author) of
upcoming books.  \camera{The following query selects} 
all the upcoming books (newtitle, author) of
authors that have not yet received a ``bad'' rating \camera{for the same title at} all three
book retailers; $\Outrel_2$ is the
output relation:
\begin{eqnarray*}
  \Outrel_1 & \coloneqq  & \SELECT \ \textrm{aut}
  \ \FROM \ \textbf{Amaz}(\textrm{ttl},\textrm{aut},\textsf{"bad"})
  \\
  & & \WHERE 
  \ \textbf{BN}(\textrm{ttl},\textrm{aut}, \textsf{"bad"})\ \AND
  \ \textbf{BD}(\textrm{ttl},\textrm{aut},\textsf{"bad"});\\
  \Outrel_2 & \coloneqq & \SELECT \ (\textrm{new},\textrm{aut})\ \FROM\
  \textbf{Upcoming}(\textrm{new},\textrm{aut}) \\
  & & \WHERE\ \NOT\ \Outrel_1(\textrm{aut});
\end{eqnarray*}
Note that this query cannot be written as a basic SGF query, since the
atoms in the query computing $\Outrel_1$ must share the $\textrm{ttl}$
variable, which is not \camera{present} in the guard of the query computing $\Outrel_2$.
\hfill $\ourend$
\end{example}


\begin{remark}\label{rem:strict:extension}
The syntax we use here differs from
the traditional syntax of the Guarded Fragment~\cite{FlumFG02}, and
is actually closer in spirit to join trees for acyclic conjunctive
queries~\cite{BernsteinC81,DBLP:journals/is/BernsteinG81},
although we do allow disjunction and negation in the where
clause. 
In the traditional syntax, a projection in the guarded fragment is
only allowed in the form $\exists \vw R(\vx) \wedge \varphi(\vz)$
where all variables in $\vz$ must occur in $\vx$. One can obtain a
query in the traditional syntax of the guarded fragment from our
syntax by adding extra projections for the atoms in $C$. For example,
$$\SELECT \ x\ \FROM \ R(x,y)\ \WHERE\ S(x,z_1)\ \AND\ \NOT\ S(y,z_2)$$
becomes $\exists y (R(x,y) \wedge (\exists z_1) S(x,z_1) \wedge 
\neg (\exists z_2) S(y,z_2))$. We note that this transformation increases the nesting depth of the query. \hfill $\ourend$
\end{remark}

\subsection{MapReduce}
\label{subsec:MR}

\newcommand{\hdfs}{HDFS\xspace}

We briefly recall the Map/Reduce model of computation (\mr for short),
and its execution in the open-source Hadoop framework~\cite{Dean:2008:MSD:1327452.1327492,White15}. 
An \emph{\mr job} is a pair $(\mu,\rho)$ of functions, where $\mu$ is
called the map and $\rho$ the reduce function. The
execution of an \mr job on an input dataset $I$ proceeds in two
stages. In the first stage, called the \emph{map stage}, each fact $f
\in I$ is processed by $\mu$, generating a collection $\mu(f)$ of
key-value pairs of the form $\kv{k}{v}$. The total collection
$\bigcup_{f \in I} \mu(f)$ of key-value pairs generated during the map
phase is then grouped on the key, resulting in a number of groups, say
$\kv{k_1}{V_1}, \dots, \kv{k_n}{V_n}$ where each $V_i$ is a set of
values. Each group $\kv{k_i}{V_i}$ is then processed by the reduce
function $\rho$ resulting again in a collection of key-value pairs
per group. The
total collection $\bigcup_{i} \rho(\kv{k_i}{V_i})$ is the output of
the \mr job.

An \emph{\mr program} is a directed acyclic graph
of \mr jobs, where an edge from job $(\mu,\rho) \to (\mu',\rho')$
indicates that $(\mu', \rho')$ operates on the output of
$(\mu,\rho)$. 
We refer to the length of the longest path in an \mr
program as the \emph{number of rounds} of the program.



\subsection{Cost Model for MapReduce}
\label{subsec:cost-model}

\begin{figure}[tbp]
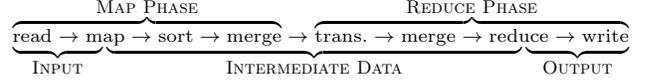

  \centering
  \scriptsize
  \[ \rlap{$\, \underbrace{\phantom{\text{reap} \rightarrow \text{m}}}_{\textsc{\scriptsize Input}} \underbrace{\phantom{\text{ap} \rightarrow \text{sort} \rightarrow \text{merge} \rightarrow \text{trans.} \rightarrow \text{merge} \rightarrow \text{red}}}_{\textsc{\scriptsize Intermediate Data
  }}\underbrace{\phantom{\text{uce} \rightarrow \text{priti}}}_{\textsc{\scriptsize Output}}$}\overbrace{\text{read} \rightarrow \text{map} \rightarrow \text{sort} \rightarrow \text{merge}}^{\textsc{\scriptsize Map Phase}} \rightarrow \overbrace{\text{trans.} \rightarrow \text{merge} \rightarrow \text{reduce} \rightarrow \text{write}}^{\textsc{\scriptsize Reduce Phase}}
  \]
  \caption{A depiction of the inner workings of Hadoop MR.}
  \label{fig:pipeline}
\end{figure}



\camera{As our aim is to reduce the total cost of parallel query plans,
we need a cost model that estimates this metric for a given \mr job.}
We briefly touch upon a cost model for analyzing the I/O complexity of an MR job based on the one introduced in~\cite{wang2013,nykiel2010} but with a  distinctive difference. The adaptation we introduce, and that is elaborated upon below, takes into account that the map function may have a different input/output ratio for different parts of the input data. 
While conceptually an MR job consists of only the map and reduce stage, its inner workings are more intricate.
Figure~\ref{fig:pipeline} summarizes the steps in the execution of an \mr job.
See~\cite[Figure~7-4, Chapter~7]{White15} for more details. 
The map phase involves 
(\textit{i}) applying the map function on the input;
(\textit{ii}) sorting and merging the {\em local} key-value pairs produced by the map function,
and (\textit{iii}) writing the result to local disk.

Let $\cI_1\cup\cdots\cup\cI_k$ denote 
the partition of the input tuples such that
the mapper behaves \emph{uniformly}\footnote{
Uniform behaviour means that for every $I_i$, each input tuple in $I_i$ is subjected to the same map function
and generates the same number of key-value pairs. In general, a partition is a subset of an input relation.}
on every data item in $\cI_i$. 
Let $N_i$ be the size (in MB) of $\cI_i$, and let $M_i$ be the size (in MB) of the intermediate data
output by the mapper on $\cI_i$.
The cost of the map phase on $\cI_i$ is:
\begin{eqnarray*}
\cost_{\map}(N_i,M_i) & = & \Hrc N_{i} + \mapmerge(M_{i}) + \Lwc M_{i},
\end{eqnarray*}
where $\mapmerge(M_i)$, denoting the cost of sort and merge in the map stage,
is expressed by
\begin{eqnarray*}
\mapmerge(M_i) & = & (\Lrc + \Lwc) M_i   \log_{D} 
\ceil[\bigg]{\dfrac{\nicefrac{(M_i + \MoutM)}{\mapn}}{\mbuflim}}.
\end{eqnarray*}
See Table~\ref{tab:costconstants} for the meaning 
of the variables $\Hrc$, $\Lwc$, $\Lrc$, $\Lwc$, $D$, $\MoutM$, $\mapn$, and $\mbuflim$.\footnote{In Hadoop 
each tuple output by the map function requires 16 bytes of metadata.}
The total cost incurred in the map phase equals the sum
\begin{equation}
\label{eq:map-cost}
\sum_{i=1}^k \ \cost_{\map}(N_i,M_i).
\end{equation}
{Note that the cost model in~\cite{wang2013,nykiel2010} defines
the total cost incurred in the map phase as}
\begin{equation}
\label{eq:map-cost-wang}
\cost_{\map}\left(\sum_{i=1}^k N_i,\sum_{i=1}^k M_i\right).
\end{equation}
The latter is not always accurate. 
Indeed, consider for instance an MR job whose input consists of two relations $R$ and $S$
where the map function outputs many key-value pairs for each tuple in $R$ \camera{and at most one} key-value pair for each tuple in $S$, e.g., because of filtering. This difference in map output may lead
to a non-proportional contribution of both \camera{input relations} 
 to the total cost.
Hence, as shown by Equation~\eqref{eq:map-cost}, 
we opt to consider different inputs separately.
This cannot be captured by map cost calculation of Equation~\eqref{eq:map-cost-wang},
as it considers the \emph{global} average map output size in the calculation
of the merge cost.
In Section~\ref{sec:experiments}, we illustrate this problem
by means of an experiment that confirms the effectiveness 
of the proposed adjustment.


To analyze the cost in the reduce phase,
let $M = \sum_{i=1}^k M_i$.
The reduce stage involves 
(\textit{i}) transferring the intermediate data (i.e., the output of the map function)
to the correct reducer,
(\textit{ii}) merging the key-value pairs locally for each reducer,
(\textit{iii}) applying the reduce function, and 
(\textit{iv}) writing the output to hdfs.
Its cost will be
\begin{eqnarray*}
\cost_{\red}(M,K) & = & \Trc M  +
 \redmerge(M)+ 
 \Hwc K,
\end{eqnarray*}
where $K$ is the size of the output of the reduce function (in MB).
The cost of merging equals
\begin{eqnarray*}
\redmerge(M) &=  (\Lrc + \Lwc) M  \log_{D} \ceil[\bigg]{\dfrac{\nicefrac{M}{\redn}}{\rbuflim}}.
\end{eqnarray*}
The total cost of an MR job equals the sum
$$
\cost_{\mystart} + \sum_{i=1}^k \cost_{\map}(N_i,M_i) + 
\cost_{\red}(M,K),
$$
where $\cost_{\mystart}$ is 
the overhead cost of starting an MR job. 


\begin{table}
\centering
\begin{tabular}{|c|l|}
\hline
$\Lrc$    & local disk read cost (per MB)\\\hline	
$\Lwc$    & local disk write cost (per MB)\\\hline
$\Hrc$    & hdfs read cost (per MB)\\\hline		
$\Hwc$    & hdfs write cost (per MB)\\\hline	
$\Trc$    & transfer cost  (per MB)\\\hline	\hline	
$\MoutM$    & map output meta-data for $\cI_i$ (in MB) \\\hline
$\mapn$   & number of mappers for $\cI_i$\\\hline
$\redn$   & number of reducers		\\\hline\hline
$D$ & external sort merge factor\\\hline
$\mbuflim$ & map task buffer limit (in MB) \\\hline
$\rbuflim$ & reduce task buffer limit (in MB) \\\hline	
\end{tabular}
\caption{Description of constants used in the cost model.
}
\label{tab:costconstants}
\end{table}


\section{Evaluating multi-semi-join and SGF Queries}
\label{sec:eval}
In this section, we describe how SGF queries can be evaluated. 
We start by introducing some necessary building blocks in
Sections~\ref{algo:semijoin} to \ref{algo:boolean}, and describe 
the evaluation of BSGF queries and multiple BSGF queries in Section~\ref{sec:basic} and \ref{sec:mul-bsgf}, respectively. These are then generalized to 
the full fragment of SGF queries in Section~\ref{sec:sgf-queries}
and \ref{sec:mul-sgf}.

First, we introduce some additional notation. 
We say that a tuple $\va=(a_1,\ldots,a_n) \in \fD^n$ of $n$
data values \emph{conforms} to a vector $\vt = (t_1,\ldots,t_n)$ of
terms, if
\begin{compactenum}
\item $\forall i,j \in [1,n]$, $t_i = t_j$ implies $a_i = a_j$; and,
\item $\forall i \in [1,n]$ if $t_i \in \fD$, then $t_i = a_i$.  
\end{compactenum}
For
instance, $(1,2,1,3)$ conforms to $(x,2,x,y)$. 
Likewise, a fact
$T(\va)$ conforms to an atom $U(\vt)$ if $T = U$ and $\va$ conforms to
$\vt$. We write $T(\va) \conforms U(\vt)$ to denote that $T(\va)$
conforms to $U(\vt)$.
If $f= R(\va)$ is a fact conforming to an atom $\alpha=R(\vt)$ and $\vx$ is a sequence of variables that occur in $\vt$, then the projection
$\pi_{\alpha;\vx} (f)$ of $f$ onto  $\vx$ is the tuple $\vb$ obtained by projecting $\va$ on the coordinates in $\vx$. For instance, 
let $f=R(1,2,1,3)$ and $\alpha = R(x,y,x,z)$. Then,
$R(1,2,1,3) \conforms R(x,y,x,z)$ and hence $\pi_{\alpha;x,z} (f) = (1,3)$.


\subsection{Evaluating One Semi-Join}
\label{algo:semijoin}
As a warm-up, let us explain how single semi-joins can be evaluated in
MR. A single semi-join is a query of the form
\begin{equation}
\label{query:bsgf}
\outrel := \SELECT\ \vw\ \FROM\ \alpha\ \WHERE\ \kappa; 
\end{equation}
where both $\alpha$ and $\kappa$ are atoms.  For notational convenience, we
will denote this query simply by $\semi{\vw}{\alpha}{\kappa}$.

To evaluate~\eqref{query:bsgf}, one can use the following one round 
repartition join~\cite{Blanas:2010bj}.  The mapper distinguishes between guard facts 
(i.e.,\ facts in $\db$ conforming to $\alpha$) and \guarded facts
(i.e.,\ facts in $\db$ conforming to $\kappa$).  Specifically, let
$\vz$ be the join key, i.e., those variables occurring in both
$\alpha$ and $\kappa$.  For each guard fact $f$ such that $f \conforms
\alpha $, the mapper emits the key-value pair
$\kv{\pi_{\alpha;\vz}(f)}{[\REQ\, \kappa; \OUT\, \pi_{\alpha;\vw}(f)]}$. Intuitively,
this pair is a ``message'' sent by guard fact $f$ to request whether a
conditional fact $g \conforms \kappa$ {with $\pi_{\kappa;\vz}(g) =
\pi_{\alpha;\vz}(f)$} exists in the database, stating that if such a
conditional fact exists, the tuple $\pi_{\alpha;\vw}(f)$ should be
output. Conversely, for each \guarded fact $g \conforms \kappa$, the
mapper emits a message of the form {$\kv{\pi_{\kappa;\vz}(g)}{[\ASSERT\,
  \kappa]}$}, asserting the existence of a $\kappa$-conforming fact in
the database with join key $\pi_{\kappa;\vz}(g)$.  On input $\kv{\vb}{V}$,
the reducer outputs all tuples $\va$ to relation $\outrel$
for which $[\REQ\,\kappa; \OUT\,\va] \in V$, provided that $V$ contains at least one assert message.

\begin{example}
Consider the query $\Outrel \coloneqq \pi_x (R(x,z)\semijoin S(z,y))$ and let $I$ contain the facts $\{R(1,2),R(4,5),S(2,3)\}$. Then the mapper emits key-value pairs
\camera{ \kv{2}{[\REQ\,S(z,y);\OUT\, 1]},
\kv{5}{[\REQ\,S(z,y);\OUT\, 4]} and,
\kv{2}{[\ASSERT\, S(z,y)]},
which after reshuffling result in groups
\kv{5}{\{[\REQ\, S(z,y);\OUT\, 4]\}} and,
\kv{2}{\{[\REQ\, S(z,y);\OUT\, 1],[\ASSERT\, S(z,y)]\}}.}
Only the reducer processing the second group produces an output, namely the fact $\Outrel(1)$.
\qquad\hfill\ourend
\end{example}

\paragraph*{Cost Analysis} 

To compare the cost of separate and combined 
evaluation of multiple semi-joins in the next section, 
we first illustrate how to analyze the cost of 
evaluating a single semi-join using the cost model
described above.
Hereto, let $|\alpha|$ and $|\kappa|$ denote the total size 
of all facts that conform to $\alpha$ and $\kappa$, respectively.
Five values are required for estimating the total cost:
$N_1, N_2, M_1, M_2$ and $K$. We can now choose
$M_1=|\alpha|$ and $M_2=|\kappa|$. For simplicity, 
we assume that key-value pairs output by the mapper 
have the same size as their corresponding input tuples, i.e., 
$N_1=M_1$ and $N_2=M_2$.\footnote{\camera{Gumbo
uses sampling to estimate $M_i$  (cf.\ Section~\ref{sec:gumbo}).}} 
Finally, the output size $K$ 
can be approximated by its upper bound $N_1$.
Correct values for meta-data size and number of mappers 
can be derived from the number of input records 
and the system settings.

\subsection{Evaluating a Collection of Semi-Joins}
\label{algo:validate}
Since a \camera{BSGF} query is \camera{essentially} a Boolean combination of semi-joins, it can be
computed by first evaluating all semi-joins followed by the evaluation
of the Boolean combination. In the present section, we introduce a
single-job MR program \validate that evaluates a set of semi-joins in
parallel. In the next section we introduce the single-job MR program
\eval to evaluate the Boolean combination.

We introduce a unary multi-semi-join operator $\msj({\cal S})$ that
takes as input a set of equations ${\cal S}=\{X_1 \coloneqq
\semi{\vx_1}{\alpha_1}{\kappa_1}, \dots, X_n \coloneqq
\semi{\vx_n}{\alpha_n}{\kappa_n}\}$. It is required that the $X_i$ are
all pairwise distinct and that they do not occur in any of the
right-hand sides. The semantics is straightforward: the operator
computes every semi-join $\semi{\vx_i}{\alpha_i}{\kappa_i}$ in ${\cal
  S}$ and stores the result in the corresponding output relation
$X_i$.

We now expand the MR job described in Section~\ref{algo:semijoin} into a job that computes $\msj({\cal S})$ by evaluating all semi-joins in parallel. Let $\vz_i$ be the join key of
semi-join $\semi{\vx_i}{\alpha_i}{\kappa_i}$.
Algorithm~\ref{alg:validate} shows the single MR job $\validate({\cal S})$ that
evaluates all $n$ semi-joins at once. 
More specifically, $\validate$
simulates the repartition join of Section~\ref{algo:semijoin}, but
outputs request messages for \emph{all} the guard facts 
at once (i.e.,
those facts conforming to one of the $\alpha_i$ for $i \in [1,n]$).
Similarly, assert messages are generated simultaneously for all of the \guarded facts (i.e., those facts conforming to one of the $\kappa_i$ for $i \in
[1,n]$). The reducer then reconciles the messages concerning the same $\kappa_i$. That is, on input $\kvpair{\vb}{V}$, the reducer outputs
the tuple $\va$ to relation $X_i$ for which {$[\REQ\,(\kappa,i); \OUT\,\va] \in V$}, provided that $V$ contains at least one message
of the form $[\ASSERT\, \kappa]$. The output therefore consists of the relations $X_1,\dots, X_n$, with each $X_i$ containing the result of evaluating
$\semi{\vx_i}{\alpha_i}{\kappa_i}$. 

Combining the evaluation of a collection of semi-joins into a single \validate\ job avoids the overhead of starting multiple jobs, reads every input relation only once, and can reduce the amount of communication by packing similar
messages together (cf.\ Section~\ref{sec:gumbo}). At the
same time, grouping all semi-joins together can potentially increase
the {average load of map and/or reduce tasks, which directly leads to
  an increased net time}. {These trade-offs are made more apparent in the following analysis and are taken into account in the algorithm \gbsgf introduced in Section~\ref{sec:basic}.}


\begin{algorithm}[t]
\caption{\validatenospace($X_1 \coloneqq \semi{\vx_1}{\alpha_1}{\kappa_1}, \dots,
   X_n \coloneqq \semi{\vx_n}{\alpha_n}{\kappa_n}$)
  }
\label{alg:validate}
\begin{algorithmic}[1]
\Function{\textsc{Map}(\textsf{Fact} $f$)}{}
\State buff = [\,]
\For{every $i$ such that $f \conforms \alpha_i$}
\State buff \textsf{+=} {$\kv{\pi_{\alpha_i;\vz_i}(f)}{[\REQ\, (\kappa_i,i);\OUT\, \pi_{\alpha_i;\vx_i}(f)]}$}
\EndFor
\For{every $i$ such that $f \conforms \kappa_i$}
\State buff \textsf{+=} {$\kv{\pi_{\kappa_i;\vz_i}(f)}{[\ASSERT\, \kappa_i]}$}
\EndFor
\State emit buffer
\EndFunction
\Statex 
\Function{\textsc{Reduce}$(\kv{k}{V})$}{} 
\ForAll{ $[\REQ\, \kappa_i;\OUT\,\va ]$ in $V$} \If{$V$
  contains $[\ASSERT\, \kappa_i]$} \State add $\va$ to $X_i$
\EndIf
\EndFor
\EndFunction
\end{algorithmic}
\end{algorithm}


\paragraph*{Cost Analysis}
{
\camera{
Let all $\kappa_i$'s be different atoms and $\alpha_1=\cdots=\alpha_n=\alpha$.}
\camera{A similar analysis can be performed for other comparable scenarios.}
As before, 
we assume that the size of the key-value pair is
the same as the size of the conforming fact, and
\camera{all tuples conform to their corresponding atom.}
The cost of $\validate(\cS)$, denoted by $\cost(\cS)$, equals 
\begin{multline}
\label{eq:sce2-group}
\cost_{\mystart} + \cost_{\map}(|\alpha|,n|\alpha|)+\sum_{i=1}^n \cost_{\map}(|\kappa_i|,|\kappa_i|) \\ + \cost_{\red}\Big( n|\alpha|+\sum_{i=1}^n |\kappa_i|,\sum_{i=1}^n |X_i|\Big),
\end{multline}
where $|X_i|$ is the size of the output relation $X_i$.
If we evaluate each $X_i$ in a separate MR job,
the total cost is:
\begin{equation}
\label{eq:sce2-sep}
\sum_{i=1}^n
\left(\!
\begin{array}{l}
\cost_{\mystart} + \cost_{\map}(|\alpha|,|\alpha|)
+\cost_{\map}(|\kappa_i|,|\kappa_i|)
\\
+\cost_{\red}(|\alpha|+|\kappa_i|,|X_i|)
\end{array}\!
\right)
\end{equation}
\camera{So, single-job evaluation of all $X_i$'s is
more efficient than separate evaluation iff
Equation~\eqref{eq:sce2-group} is less than
Equation~\eqref{eq:sce2-sep}.}}



\subsection{Evaluating Boolean Combinations}
\label{algo:boolean}


Let \camera{$X_0,X_1, \dots X_n$} be relations with the same arity and let
$\varphi$ be a Boolean formula over $X_1, \dots X_n$. 
It is straightforward to evaluate \camera{$X_0 \land \varphi$} in a single MR job: on each fact $X_i(\va)$, 
the mapper emits $\kv{\va}{i}$. The reducer hence receives pairs $\kv{\va}{V}$ with $V$ containing all the indices $i$ for which $\va \in X_i$, and
outputs $\va$ only if the Boolean formula, obtained from \camera{$X_0 \land \varphi$} by replacing
every $X_i$ with {\em true} if $i\in V$ and {\em false} otherwise, evaluates to true. 
For instance, if $\varphi = X_1 \land X_2 \land \neg X_3$, it will emit $\va$ only if $V$ contains \camera{0}, $1$ and $2$ but not $3$.

We denote this MR job as \camera{$\eval(X_0, \varphi)$}. {We emphasize that multiple
Boolean formulas \camera{$Y_1 \land \varphi_1,\ldots,Y_n \land \varphi_n$} with distinct sets of variables can be evaluated 
\camera{in one MR job}
which we denote as \camera{$\eval(Y_1, \varphi_1,\ldots, Y_n, \varphi_n)$}.}

\paragraph*{Cost Analysis}
Let $|X_i|$ be the size of relation $X_i$.
Then, when $|\varphi|$ is the size of the output,
$\cost(\camera{\eval(X_0,\varphi)}) $ equals
\begin{equation}
\cost_{\mystart} + \sum_{i=0}^n \cost_{\map}(|X_i|,|X_i|)
+ \cost_{\red}\big(\sum_{i=0}^n |X_i|,|\varphi|\big).
\end{equation}

\subsection{Evaluating BSGF Queries}
\label{sec:basic}
\label{sec:costmodel}

\begin{figure}
  \centering
  \begin{subfigure}[b]{0.5\linewidth}
    \centering
\scalebox{0.8}{
\includegraphics{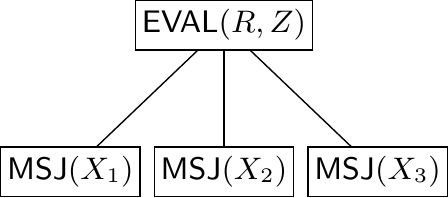}
}
  \caption{}
  \end{subfigure}

   \begin{subfigure}[b]{0.45\linewidth}
     \centering
\scalebox{0.8}{
\includegraphics{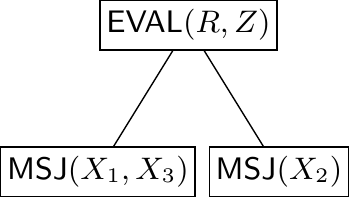}
}
  \caption{}
  \end{subfigure}
  \qquad
 \begin{subfigure}[b]{0.45\linewidth}
   \centering
\scalebox{0.8}{
\includegraphics{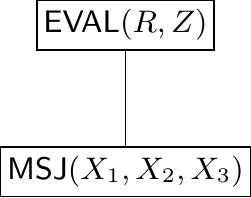}
}
  \caption{}
  \end{subfigure}

  \caption{Different possible query plans for the query given in Example~\ref{ex-plan-query}. \camera{Here, $X_1\coloneqq R(x,y) \semijoin S(x,z)$, $X_2\coloneqq R(x,y) \semijoin T(y)$, $X_3\coloneqq R(x,y) \semijoin  U(x)$ and $Z \coloneqq X_1 \land (X_2 \lor \neg X_3)$; trivial projections are omitted.} 
  }
  \label{fig:plans}
\end{figure}
We now have the building blocks to discuss the evaluation of basic
queries.  Consider the following basic query $Q$: 
\[ \Outrel \coloneqq \SELECT \ \vw\ \FROM \ R(\vt)\ \WHERE \ C. \qquad
\] Here, $C$ is a Boolean combination of \guarded atoms
$\kappa_i$, for $i\in [1,n]$, that can only share variables occurring
in $\vt$. 
Note that it is implicit that $\kappa_1,\ldots,\kappa_n$
are all different atoms.
Furthermore, let ${\cal S}$ be
the set of equations $\{X_1 \coloneqq \semi{\vw}{R(\vt)}{\kappa_1},
\dots, X_n \coloneqq \semi{\vw}{R(\vt)}{\kappa_n}\}$ and let $\varphi_C$
be the Boolean formula obtained from $C$ by replacing every
conditional atom $\kappa_i$ by $X_i$.

Then, for every partition $\{\cS_1,\ldots,\cS_p\}$ of $\cS$, the following MR program 
computes $Q$: 
\begin{center}
\includegraphics{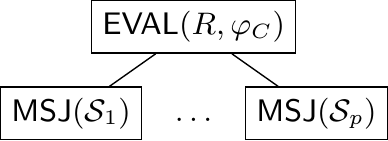}
\end{center}
We refer to any such program as a \emph{basic MR program for $Q$}.
Notice that all \validate jobs can be executed in parallel. 
So, the above program consists in fact of two rounds,
but note that there are $p+1$ MR jobs in total:
one for each $\validate(\cS_i)$, and one for $\eval(R,\varphi_C)$.



\begin{example}
\label{ex-plan-query}
\camera{Figure~\ref{fig:plans} shows three alternative basic MR programs for the
following query:}
\begin{eqnarray}
  \label{eq:ex-plan-query}
  \Outrel & \coloneqq & \SELECT \ x,y\ \FROM \ R(x,y) \nonumber \\ && \WHERE\ S(x,z)\ \AND\ (T(y)\ \OR\
\NOT\ U(x))
\end{eqnarray}
In alternative (a), all semijoins $X_1,X_2,X_3$ 
are evaluated as separate jobs. 
In alternative (b), $X_1$ and $X_3$ are computed in one job,
while $X_2$ is computed separately. 
In alternative (c), all semijoins $X_1,X_2,X_3$
are computed in a single job.~\hfill$\ourend$
\end{example}

\paragraph*{Cost Analysis}
When $\cS$ is partitioned into $\cS_1\cup\cdots\cup\cS_p$,
the cost of the MR program is:
\begin{equation}
\label{eq:cost-bsgf}
\camera{\cost(\eval(R,\varphi_C))} + 
\sum_{i=1}^p \cost(\cS_i),
\end{equation}
{where the cost of $\cost(\cS_i)$ 
is as in Equation~(\ref{eq:sce2-group}).}

\paragraph*{Computing the Optimal Partition}
By \bsgfopt we denote the problem that takes a BSGF query $Q$ as above and
computes a partition $\cS_1\cup \cdots\cup\cS_p$ of $\cS$
such that its total cost as computed in Equation~(\ref{eq:cost-bsgf}) is {\em minimal}. 
\camera{The \textsc{Scan-Shared Optimal Grouping} problem,
which is known to be NP-hard,
is reducible to this problem~\cite{nykiel2010}:}
\begin{theorem}
\label{theorem:bsgfopt-np}
The decision variant of\/ \bsgfopt is NP-complete.
\end{theorem}


While for small queries the optimal solution can be found using a brute-force
search, for larger queries we adopt the fast greedy heuristic introduced by Wang et al\cite{wang2013}.
For two disjoint subsets $\cS_i,\cS_j \subseteq \cS$,
define:
\[
\gain(\cS_i,\cS_j) = \cost(\cS_i) + \cost(\cS_j) - \cost(\cS_i\cup\cS_j).
\]
That is, $\gain(\cS_i,\cS_j)$ denotes the cost gained
by evaluating $\cS_i\cup\cS_j$ in one MR job
rather than evaluating each of them separately.
For a partition $\cS_1\cup\cdots\cup\cS_p$,
our heuristic algorithm greedily finds a pair $i,j \in [p]\times[p]$
such that $i\neq j$ and $\gain(\cS_i,\cS_j) > 0$ is the greatest.
If there is such a pair $i,j$,
we merge $\cS_i$ and $\cS_j$ into one set.
We iterate such heuristic 
starting with the trivial partition
$\cS_1\cup\cdots\cup\cS_n$, where each 
$\cS_i = \{X_i \coloneqq \semi{\vw}{R(\vt)}{\kappa_i}\}$.
The algorithm stops when there is no pair $i,j$ for which
$\gain(\cS_i,\cS_j) > 0$.
We refer to this algorithm as \gbsgf.
For a BSGF query $Q$, 
we denote by $\opt(Q)$ 
the optimal (least cost) basic MR program for $Q$,
and by $\gopt(Q)$  we denote the program computed by \gbsgf.

%
%

\subsection{Evaluating Multiple BSGF Queries}
\label{sec:mul-bsgf}

The approach presented in the previous section can be readily adapted to evaluate multiple BSGF queries. Indeed, consider a set of $n$ BSGF queries, each of the form
\[ \Outrel_i \coloneqq \SELECT \ \vw_i\ \FROM \
R_i(\vt_i)\ \WHERE \ C_i \] 
where none of the $C_i$ can refer to any of the $\Outrel_j$.
A corresponding MR program is then of the form
\begin{center}
\includegraphics{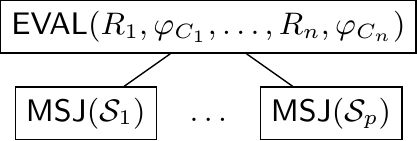}
\end{center}
The child nodes constitute a partition of all the necessary semi-joins.
Again, $\varphi_{C_i}$ is the Boolean formula obtained from $C_i$.
We assume that the set of variables used in the Boolean formulas are disjoint.
For a set of BSGF queries $F$, we refer to any MR program of the above form
as a \emph{basic MR program for $F$},
whose cost can be computed in a similar manner as \camera{above}.
The optimal basic program for $F$ and the program computed by the greedy algorithm of Section~\ref{sec:basic} are denoted by $\opt(F)$ and $\gopt(F)$, respectively,
and their costs are denoted by $\cost(\opt(F))$ and $\cost(\gopt(F))$.


\subsection{Evaluating SGF Queries}
\label{sec:sgf-queries}

Next, we turn to the evaluation of SGF queries. Recall that 
an SGF query $Q$ is a sequence of basic queries of the form
\camera{$
\Outrel_1  \coloneqq  \bodyquery_1; \ldots; \Outrel_n  \coloneqq \bodyquery_n;
$}
where each $\bodyquery_i$ can refer to the relations $\Outrel_j$
with $j<i$. We denote the BSGF $Z_i\coloneqq \xi_i$ by $Q_i$.
The most naive way to compute $Q$ is to evaluate the \camera{BSGF queries in} $Q$ sequentially,
where each $\bodyquery_i$ is evaluated  using the approach detailed in the 
previous section.
This leads to a $2n$-round MR program. 
We would like to have a better strategy that aims at decreasing the total time
by combining the evaluation of different independent subqueries. 

To this end, let $\cG_Q$ be the dependency graph induced by $Q$. That is,
$\cG_Q$ consists \camera{of a set $F$} of $n$ nodes (one for each BSGF query)
and there is an edge from $Q_i$ to $Q_j$ if relation $\Outrel_i$ is mentioned 
in $\bodyquery_j$.
A \emph{multiway topological sort} of the dependency graph
$\cG_Q$ is a sequence $(F_1,\dots,F_k)$ such that
\begin{compactenum}
\item 
$\{F_1,\dots,F_k\}$ is a partition of \camera{$F$};
\item 
if there is an edge from node $u$ to node $v$ in $\cG_Q$, 
then $u \in F_i$ and $v \in F_j$ such that $i < j$.
\end{compactenum}
Notice that any multiway topological sort $(F_1,\ldots,F_k)$ of $\cG_Q$ provides a valid \camera{ordering} to evaluate $Q$,
i.e., all the queries in $F_i$ are evaluated before $F_j$ whenever $i < j$.

\begin{example}
\label{ex:sgf}
Let us illustrate the latter by means of an example.
Consider the following SGF query $Q$:
 \begin{eqnarray}
   Q_1: \quad \Outrel_1 & := & \SELECT \ x,y\ \FROM \ R_1(x,y)\ \WHERE\ S(x) \nonumber\\
   Q_2: \quad \Outrel_2 & := & \SELECT \ x,y\ \FROM \ \Outrel_1(x,y)\ \WHERE\ T(x)\nonumber\\
   Q_3: \quad \Outrel_3 & := & \SELECT \ x,y\ \FROM \ \Outrel_2(x,y)\ \WHERE\ U(x)\nonumber\\
   Q_4:  \quad \Outrel_4 & := & \SELECT \ x,y\ \FROM \ R_2(x,y)\ \WHERE\ T(x)\nonumber\\
   Q_5: \quad \Outrel_5 & := & \SELECT \ x,y\ \FROM \ \Outrel_3(x,y)\ \WHERE\ \Outrel_4(x,x)\nonumber
 \end{eqnarray}
The dependency graph $\cG_Q$ is as follows:
\begin{center}
\includegraphics{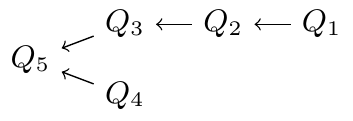}
\end{center}
There are four possible multiway topological sorts of $\cG_Q$:
\begin{compactenum}
\item $(\{Q_1, Q_4\},\ \{Q_2\}, \ \{Q_3\}, \ \{Q_5\})$.
\item $(\{Q_1\},\ \{Q_2,Q_4\}, \ \{Q_3\}, \ \{Q_5\})$.
\item $(\{Q_1\},\ \{Q_2\},\ \{Q_3,Q_4\}, \ \{Q_5\})$.
\item $(\{Q_1\},\ \{Q_2\},\ \{Q_3\}, \ \{Q_4\}, \{Q_5\})$.\hfill $\ourend$
\end{compactenum}


\end{example}



Let $\cF =(F_1,\ldots,F_k)$ be a topological sort of $\cG_Q$.
Since the optimal program $\opt(F_i)$, defined in Subsection~\ref{sec:mul-bsgf},
is intractable (due to Theorem~\ref{theorem:bsgfopt-np}),
we will use the greedy approach to evaluate $F_i$, i.e.,
$\gopt(F_i)$ as defined in Section~\ref{sec:mul-bsgf}.
The cost of evaluating $Q$ according to $\cF$ is
\begin{equation}
\label{eq:cost-sgf-gopt}
\cost(\cF) \ = \
\sum_{i=1}^k \ \cost(\gopt(F_i))
\end{equation}

We define the optimization problem \sgfopt that takes as input an SGF query $Q$
and constructs a multiway topological sort $\cF$ of $\cG_Q$ 
with minimal $\cost(\cF)$.
\camera{A reduction from {\sc Subset Sum}~\cite{GareyJohnson} yields following result (cf.\ \citefullpaperA):}

\begin{theorem}
\label{theorem2}
The decision variant of \sgfopt is NP-complete.
\end{theorem}
In the following, we present a novel heuristic 
for computing a multiway topological sort of an SGF 
that tries to maximize the overlap between queries.
To this end, we define the \emph{overlap} between a BSGF query $Q$ 
and a set of BSGF queries $F$, denoted by $\overlap(Q,F)$, 
to be the number of relations occurring in $Q$ that also occur in $F$. 
For instance, in Example~\ref{ex:sgf}, the overlap between $Q_2$ and $\{Q_1,Q_3,Q_4,Q_5\}$ 
is 1 as they share only relation $T$. 

Consider the following algorithm \gsgf\
that computes a multiway topological sort $\cF$ of \camera{an SGF query} $Q$.
\camera{
Initially, all the vertices in the dependency graph $\cG_Q$
are colored blue and $\cX = ()$.
The algorithm performs the following iteration
with the invariant that $\cX$
is a multiway topological sort of the red vertices in $\cG$:
}
\begin{compactenum}
\item
Suppose $\cX = (F_1,\ldots,F_m)$ and
\camera{blue vertices remain.}
\item
Let $D$ be the set of those blue vertices in $\cG_Q$
for which none of the incoming edges are from other blue vertices.
(Due to the acyclicity of $\cG_Q$, 
the set $D$ is non-empty if $\cG_Q$ still has blue vertices.)
\item
Find a pair $(u,F_i)$ such that $u \in D$,
$(F_1,\ldots,F_i\cup \{u\},\ldots,F_m)$ is a topological sort
of the vertices \camera{$\{u\}\cup\bigcup_i F_i$}, and
$\overlap(u,F_i)$ is non-zero.
\item
If such a pair $(u,F_i)$ exists,
choose one with maximal $\overlap(u,F_i)$,
and set $\cX = (F_1,\ldots,F_i\cup \{u\},\ldots,F_m)$.
Otherwise, set $\cX = (F_1,\ldots,F_m,\{u\})$.
\item
Color the vertex $u$ red.
\end{compactenum}
The iteration stops when every vertex in $\cG_Q$ is red,
and hence, $\cX$ is a multiway topological sort of $\cG_Q$.
Clearly, the number of iterations is $n$,
where $n$ is the number of vertices in $\cG_Q$.
Each iteration takes $O(n^2)$.
Therefore, the heuristic algorithm outlined above runs in $O(n^3)$ time.

\camera{Note that a naive dynamic evaluation strategy may consist of re-running \gsgf
after each BSGF evaluation in order to obtain an updated \mr query plan.}

\subsection{Evaluating Multiple SGF Queries}
\label{sec:mul-sgf}

Evaluating a collection of SGF queries can be done in the same way as evaluating
one SGF query. Indeed, we can simply consider the union of all BSGF subqueries.
Note that
\camera{this strategy can exploit} overlap between \camera{different sub\-queries}, potentially bringing down the total and/or net time.






\section{Experimental validation}
\label{sec:experiments}

In this section, we experimentally validate the effectiveness of our
algorithms. First, we discuss
our experimental setup in Section~\ref{sec:gumbo}.
In Section~\ref{sec:exp:bsgf}, we discuss the evaluation of BSGF queries.
In particular, we compare with Pig and Hive and address the effectiveness of 
the cost model. 
The experiments concerning nested SGF queries are presented in Section~\ref{sec:exp:sgf}.
Finally, Section~\ref{sec:exp:system} 
discusses the overal performance of our own system called \Gumbo.


\subsection{Experimental Setup}
\label{sec:gumbo}

The algorithms \gbsgf and \gsgf are implemented in a system called
\Gumbo~\cite{gumbo,DaenenNT15}. \Gumbo runs on top of Hadoop,
and adopts several important optimizations:
\begin{compactenum}[(1)]
\item Message packing, as also used in \cite{wang2013}, reduces
  network communication by packing all the request and assert messages
  associated with the same key into one list. 

\item Emitting a reference  to each guard tuple (i.e., a tuple id) rather than the tuple itself {when evaluating (B)SGF queries}
significantly reduces the number of bytes that are shuffled.  To compensate for this reduction, the guard relation needs to be re-read in the \eval job
but the latter is insignificant w.r.t. the gained improvement.

\item Setting the number of reducers in function of the intermediate
data size. An estimate of the intermediate size is obtained through
simulation of the map function on a sample of the input relations.
The latter estimates are also used as approximate values for $N_{\inp}$, $N_{\inter}$, and $N_{\out}$. For the experiments below, 256MB of data was allocated to each reducer.

\item 
{When the conditional atoms of a BSGF query 
all have the same join-key, the query can be evaluated in one job
by combining \validate and \eval.
A similar reduction to one job can be obtained when the 
the Boolean condition is restricted to 
only disjunction and negation. 
The same optimization also works for multiple BSGF queries.} We refer to these
programs as \ONEROUND below.
\end{compactenum}

All experiments are conducted on the HPC infrastructure of the
Flemish Supercomputer Center (VSC). 
Each experiment
was run on a cluster consisting of 10 compute nodes.  Each node
features two 10-core ``Ivy Bridge'' Xeon E5-2680v2 CPUs (2.8 GHz, 25
MB level 3 cache) with 64 GB of RAM and a single 250GB harddisk.  The
nodes are linked to a IB-QDR Infiniband network.  We used Hadoop~2.6.2, \camera{Pig~0.15.0 and Hive~1.2.1};
the specific Hadoop settings and cost model constants can be found in \citefullpaperB.
All experiments are run three times; average results are reported.

Queries typically contain a multitude of relations and the input sizes of our experiments go up to 100GB depending on the query {and the evaluation strategy}.
The data that is used for the guard relations consists of 
100M tuples that add up to 4GB per relation. For the conditional relations we use the same number of tuples that add up to 1GB per relation; {50\% of the conditional tuples \camera{match} those of the guard relation.}


We use the following performance metrics: 
\begin{compactenum}
\item \emph{total time}: the aggregate sum of time spent by all mappers and reducers;
\item \emph{net time}: elapsed time between query submission to obtaining the final result;
\item \emph{input cost}: the number of bytes read from hdfs over the entire MR plan; 
\item \emph{communication cost}: the number of bytes that are transferred from mappers to reducers.
\end{compactenum}

\subsection{BSGF Queries}
\label{sec:exp:seq}
\label{sec:exp:bsgf}

\begin{table}
\centering
\scriptsize
\begin{tabular}{|r|p{4cm}|p{2cm}|}
\hline
QID & Query & Type of query\\\hline
\hline
A1	&$R(x,y,z,w) \semijoin $\newline\phantom{space}$S(x) \land T(y) \land U(z) \land V(w)$ 	& guard sharing\\\hline
A2	&$R(x,y,z,w) \semijoin  $\newline\phantom{space}$S(x) \land S(y) \land S(z) \land S(w)$ 	& guard \& conditional name sharing\\\hline
A3	&$R(x,y,z,w) \semijoin  $\newline\phantom{space}$S(x) \land T(x) \land U(x) \land V(x)$	& guard \& conditional key sharing\\\hline
A4	&$R(x,y,z,w) \semijoin  $\newline\phantom{space}$S(x) \land T(y) \land U(z) \land V(w)$ \newline $G(x,y,z,w) \semijoin  $\newline\phantom{space}$W(x) \land X(y) \land Y(z) \land Z(w)$	& no sharing \\\hline
A5	&$R(x,y,z,w) \semijoin  $\newline\phantom{space}$S(x) \land T(y) \land U(z) \land V(w)$ \newline $G(x,y,z,w) \semijoin  $\newline\phantom{space}$S(x) \land T(y) \land U(z) \land V(w)$	&  conditional name sharing \\\hline\hline
B1	&$R(x,y,z,w) \semijoin   $\newline\phantom{spa}$S(x) \land T(x) \land U(x) \land V(x) \land {}$\newline\phantom{spa}$S(y) \land T(y) \land U(y) \land V(y) \land {}$\newline\phantom{spa}$S(z) \land T(z) \land U(z) \land V(z) \land {}$\newline\phantom{spa}$S(w) \land T(w) \land U(w) \land V(w)$	& large conjunctive query \\\hline
B2	&$R(x,y,z,w) \semijoin   $\newline$( S(x) \land \lnot T(x) \land \lnot U(x) \land \lnot V(x)) \lor (\lnot S(x) \land T(x) \land \lnot U(x) \land \lnot V(x)) \lor ( S(x) \land \lnot T(x) \land U(x) \land \lnot V(x)) \lor (\lnot S(x) \land \lnot T(x) \land \lnot U(x) \land V(x))$	& uniqueness query\\\hline
\end{tabular}
\caption{Queries used in the BSGF-experiment}
\label{tab:queriesbsgf}
\end{table}

\begin{figure*}
\centering
\begin{subfigure}[b]{0.47\textwidth}
\begin{center}
\includegraphics[width=\linewidth]{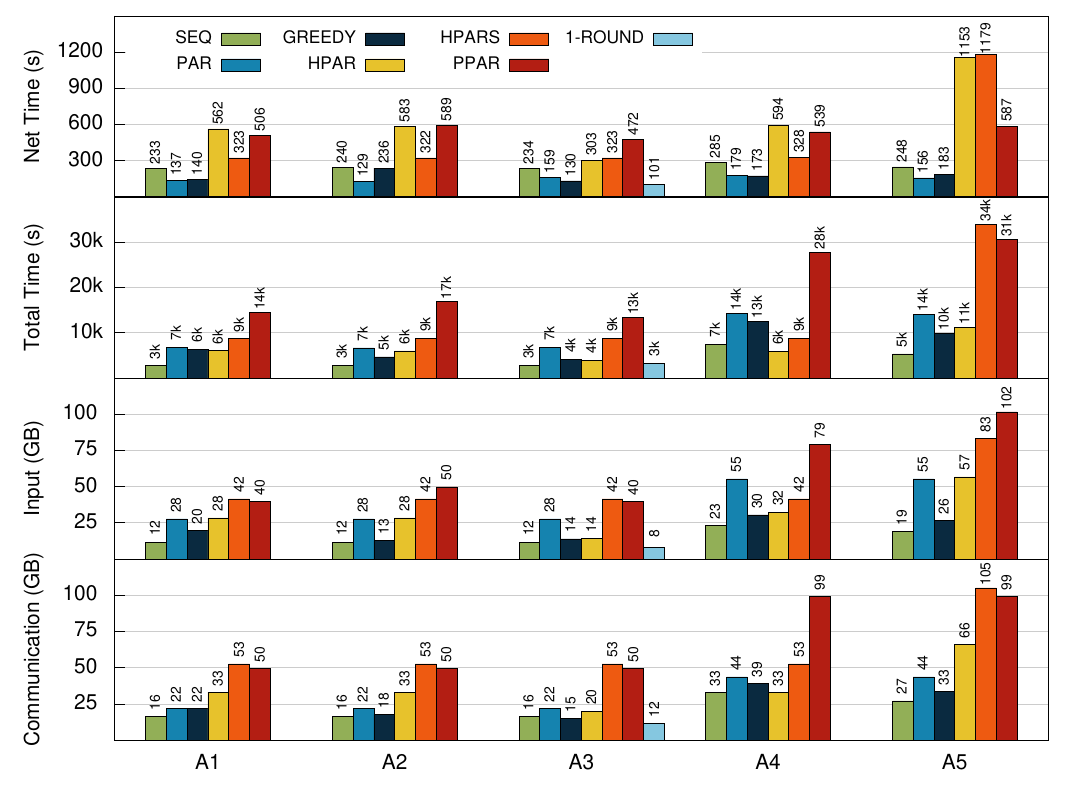}
\end{center}
 \caption{Absolute values.}
 \label{fig:bgsf:absolute}
\end{subfigure}
\qquad
\begin{subfigure}[b]{0.47\textwidth}
\begin{center}
\includegraphics[width=\linewidth]{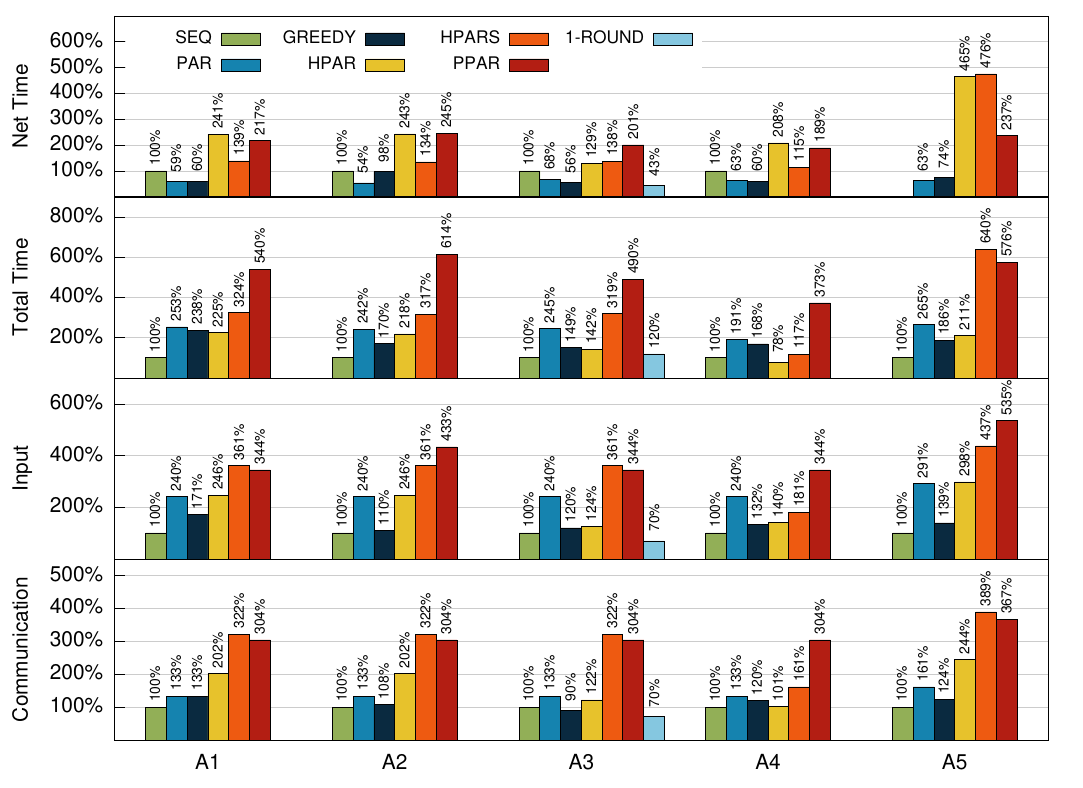}
\end{center}
 \caption{Values relative to \SEQ.}
 \label{fig:bgsf:relative}
\end{subfigure}
\caption{Results for evaluating the BSGF queries using different strategies.}
\label{fig:resultsbgsf}
\end{figure*}

Table~\ref{tab:queriesbsgf} lists the type of BSGF queries used in this section.\footnote{
The results obtained here generalize to non-conjunctive BSGF queries.
Conjunctive BSGF queries were chosen here to simplify the comparison  with sequential query plans.}
Figures~\ref{fig:resultsbgsf}~\&~\ref{fig:bgsf:b1b2} show the results that are discussed~next.

\paragraph*{Sequential vs.\ Parallel} 
We first compare sequential and parallel evaluation of queries A1--A5
to highlight the major differences between sequential 
and parallel query plans and to illustrate the effect of grouping. 
In particular, we consider three evaluation strategies in \Gumbo: 
(\textit{i}) 
evaluating all semi-joins sequentially by applying a semi-join
to the output of the previous stage (\SEQ), where the number of rounds
depends on the number of semi-joins;
(\textit{ii}) using the 2-round strategy with algorithm \gbsgf (\MSJ);
and, (\textit{iii}) a more naive version of \MSJ where 
no grouping occurs, i.e., every semi-join is evaluated separately
in parallel (\PAR). 
\camera{As semi-join algorithms in MR have not received significant attention,
we choose to compare with the two extreme approaches: no parallelization (\SEQ)
and parallelization without grouping (\PAR).}
Relative improvements of \PAR and \MSJ w.r.t.\ \SEQ 
are shown in Figure~\ref{fig:bgsf:relative}.

We find that both \PAR and \MSJ result in lower net times. 
In particular, we see average improvements of 39\%
and 31\% over \SEQ, respectively. 
On the other hand, the total times for \PAR are much higher
than for \SEQ: 132\% higher on average. This is explained by
the increase in both input and communication bytes, whereas
the data size can be reduced after each step in the sequential evaluation.
For \MSJ, total times vary depending on the structure of the query. 
Total times are significantly reduced 
for queries where conditional
atoms share join keys and/or relation names. This effect is most
obvious for queries A1, A2 and A5 {where we oberve
reductions in net time of 30\%, 29\% and 30\%, respectively,
w.r.t.\ \PAR.}

For query A3, all conditional atoms have the same join key,
making 1-round (\ONEROUND, see Section~\ref{sec:gumbo}) 
evaluation possible. This further reduces
{the total and net time to only 49\% and 63\% of those of \PAR, respectively.}




\paragraph*{Hive \& Pig}
We now {examine} parallel query evaluation in Pig and Hive
and show that \Gumbo outperforms both systems for BSGF queries.
For this test, we implement the 2-round query plans of Section~\ref{sec:basic}
directly in Pig and Hive. For Hive, we consider two evaluation strategies:
one using Hive's left-outer-join operations (\HIVEPAR) and one 
using Hive's semi-join operations (\HIVEPARSEM).
For Pig, we consider one strategy
that is implemented using the COGROUP operation (\PIGPAR).
We also studied sequential evaluation of BSGF queries in both systems
but choose to omit the results here as both performed drastically worse 
than their \Gumbo equivalent (\SEQ) in terms of net and total time. 

%

First, we find that \HIVEPAR lacks parallelization.  This is caused by
Hive's restriction that certain join operations are executed
sequentially, even when parallel execution is enabled.  This leads to
net times that are 238\% higher on average, compared to \PAR. Note
that query A3 shows a better net time than the other queries. This is
caused by Hive allowing grouping on certain join queries, effectively
bringing the number of jobs (and rounds) down to 2.

Next, we find that \HIVEPARSEM performs better than \HIVEPAR
in terms of net time but is still 126\% higher on average than \PAR.
The lower net times w.r.t.\ \HIVEPAR are explained by Hive
allowing parallel execution of semi-join operations, without allowing
any form of grouping. This effectively makes \HIVEPAR 
the Hive equivalent of \PAR. The high net times are caused by
Hive's higher average map and reduce input sizes.
 
Finally, Pig shows an average net time increase of 254\%.
This is mainly caused by the lack of reduction in intermediate
data and in input bytes, together with \camera{input-based reducer allocation}
(1GB of \camera{map} input data
per reducer). For these queries, this leads to a low number of reducers,
causing the average reduce time, and hence overall net time, 
to go up.

As the reported net times for Hive and Pig are much higher than 
for sequential evaluation in \Gumbo (\SEQ), 
we conclude that Pig and Hive, with default settings, 
are unfit for parallel \camera{evaluation} of BSGF queries. For this reason
we restrict our attention to \camera{\Gumbo} in the following sections.

\paragraph*{Large Queries}
Next, we compare the evaluation of two larger
BSGF queries B1 and B2 from Table~\ref{tab:queriesbsgf}. 
The results are shown in Figure~\ref{fig:bgsf:b1b2}.
Query B1 is a conjunctive BSGF query featuring a high 
number of atoms. Its structure ensures a deep sequential
plan that results in a high net time for \SEQ.
We find that \PAR only takes 22\% of the net time,
which shows that parallel query plans can yield significant improvements.
\camera{Conversely}, \PAR takes up 261\% more total time than \SEQ, 
as the latter is more efficient in pruning the data at each step.
Here, \MSJ is able to successfully parallelize query execution 
without sacrificing total time. Indeed, \MSJ exhibits a net time 
comparable to that of \PAR and a total time comparable to that of \SEQ. 

Query B2 consists of a large boolean combination
and is called the \emph{uniqueness query}. This query returns
the tuples that can be connected to precisely one of the
conditional relations through a given attribute. The number
of distinct conditional atoms is limited, and the disjunction at the highest
level makes it possible to evaluate the four conjunctive subexpressions
in parallel \camera{using \SEQ.}
Still, we find that the net time of of \PAR improves that of \SEQ by 
66\%. As \PAR only needs to calculate the result of four semi-join
queries in its first round, we also find a reduction of 57\% in total time.
\MSJ further reduces both numbers.

Finally, for B2, a 1-round evaluation (\ONEROUND, see Section~\ref{sec:gumbo})
can be considered, as only 
one key is used for the conditional atoms. This evaluation strategy
\camera{brings down} both net and total \camera{time} of \SEQ by more
than 80\%. 

\paragraph*{Cost Model}
As explained in Section~\ref{subsec:cost-model}, the major difference
between our cost model and that of Wang~et~al.~\cite{wang2013}
(referred to as \ourcost and \wangscost, respectively, from here onward)
concerns identifying the individual map cost contributions of the input relations. 
For queries where the map
input/output ratio differs greatly among the input relations,
we notice a vast improvement for the \GCOST strategy. 
We illustrate this using the following query:
\begin{align*}
R(x,y,z,w) {}\semijoin{} &S_1(\vx_1,c) \land \ldots \land S_1(\vx_{12},c) \land{} \\
	&S_2(\vx_1,c) \land \ldots \land S_2(\vx_{12},c) \land{}\\
	&S_3(\vx_1,c) \land \ldots \land S_3(\vx_{12},c) \land{}\\
	&S_4(\vx_1,c) \land \ldots \land S_4(\vx_{12},c),
\end{align*}
where $\vx_1,\ldots,\vx_{12}$ are all distinct keys and 
$c$ is a constant that filters out all tuples from $S_1,\ldots,S_4$.
The results for evaluating this query using \GCOST with
\ourcost and \wangscost are unmistakable:
\ourcost provides a 43\% reduction in total time and
a 71\% reduction in net time. The explanation is that
\wangscost does not discriminate between different input
relations, it averages out the intermediate data and \camera{therefore} fails
to account for the high number of map-side merges
and the accompanying increase in both total and net time. 

For queries A1--A5 and B1--B2, where input relations
have a contribution to map output that is proportional to their input size, 
we find that both cost models behave similarly. 
\camera{When comparing two random jobs, the cost models
correctly identify the highest cost job in 72.28\% (\ourcost) and
69.37\% of the cases (\wangscost).}
Hence, we find that \ourcost provides a 
more robust \camera{cost estimation
as it can isolate} input relations that have a
non-proportional contribution to the map output, 
while it automatically resorts to \wangscost
in the case of \camera{an equal} contribution.

\paragraph*{Conclusion} 
We conclude that parallel evaluation ef\-fec\-tive\-ly
lowers net times, at the cost of higher total times.
\MSJ, {backed by an updated cost model}, 
successfully manages to bring down total
times of parallel evaluation, especially in the presence of 
commonalities among the atoms of BSGF queries.
For larger queries, total times similar to \SEQ
are obtained. Finally, \Gumbo outperforms
Pig and Hive in all aspects when it comes to
parallel evaluation of BSGF queries.

\begin{figure}[t]
\centering
\begin{center}
\includegraphics[width=\linewidth]{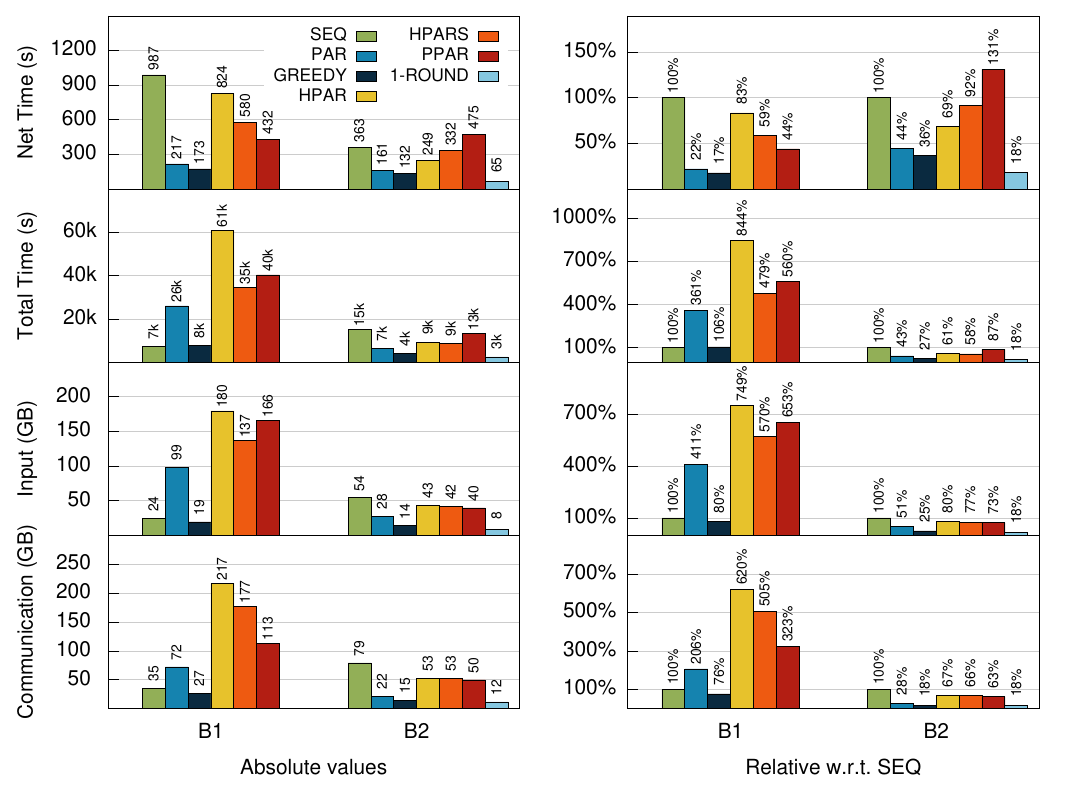}
\end{center}
 \caption{ Large BSGF queries.}
 \label{fig:bgsf:b1b2}
\end{figure}

\begin{figure}[t]
\begin{center}
\includegraphics[width=\linewidth]{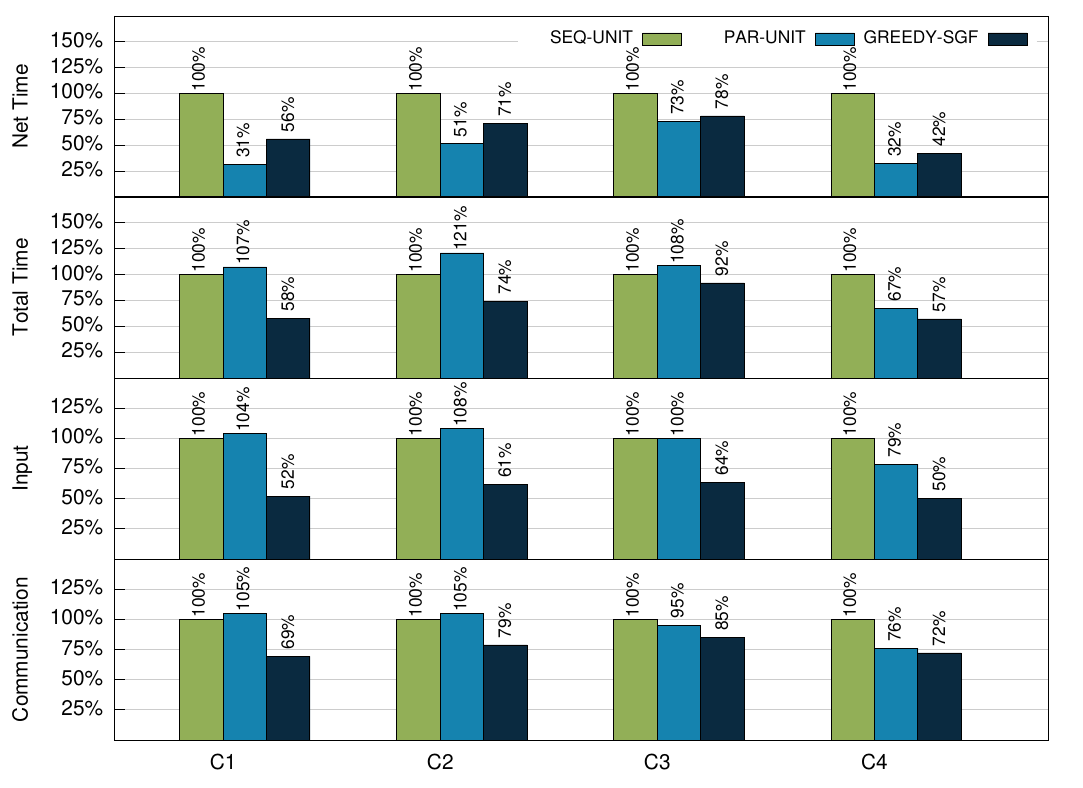}
\end{center}
\caption{SGF results, values relative to \SEQUNIT.}
\label{fig:resultssgf}
\end{figure}

\begin{figure}
  \centering
    \begin{subfigure}[b]{\linewidth}
    \centering
\scalebox{0.55}{
\includegraphics{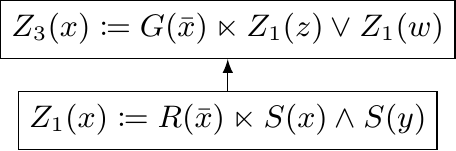}
\includegraphics{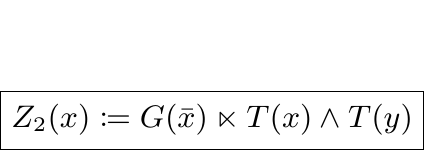}
\includegraphics{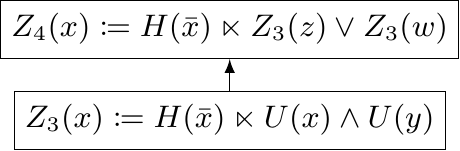}
}
  \caption{Query Set C1}
  \end{subfigure}

    \vspace{1em}
    \begin{subfigure}[b]{\linewidth}
    \centering
\scalebox{0.55}{
\includegraphics{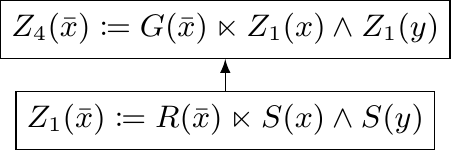}
\includegraphics{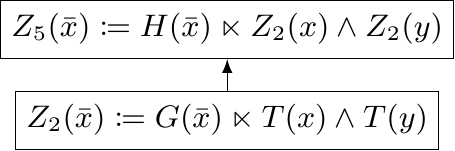}
\includegraphics{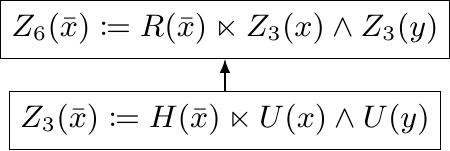}
}
  \caption{Query Set C2}
  \end{subfigure}
  
  \vspace{1em}
    \centering
    \begin{subfigure}[b]{\linewidth}
    \centering
\scalebox{0.6}{
\includegraphics{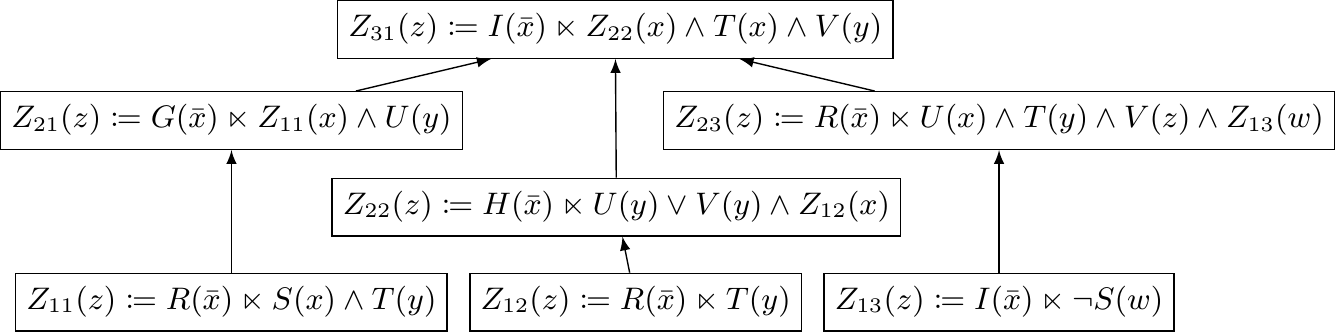}
}
  \caption{Query C3}
  \end{subfigure}
  
  \vspace{1em}
     \begin{subfigure}[b]{\linewidth}
    \centering
\scalebox{0.6}{
\includegraphics{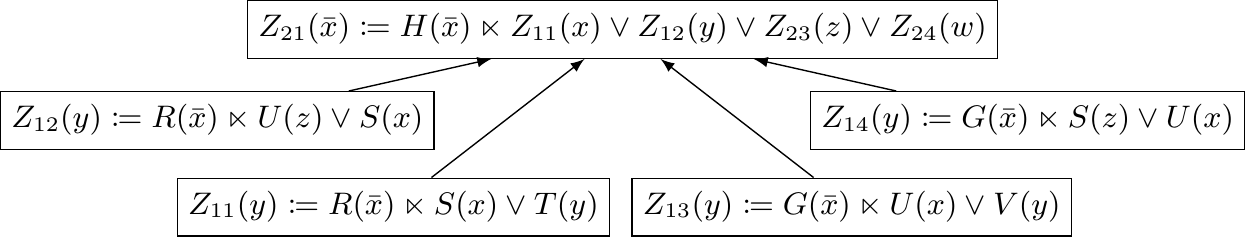}
}
  \caption{Query C4}
  \end{subfigure}
  \qquad

    \caption{The queries used in the SGF experiment. Each node represents one BSGF subquery ($\vx = x,y,z,w)$.}
  \label{fig:sgfqueries}
\end{figure}

\subsection{SGF Queries}
\label{sec:exp:sgf}
In this section, we show that the algorithm \gsgf 
succeeds in lowering total time while avoiding significant \camera{increase} 
in \camera{net time}. Figure~\ref{fig:sgfqueries} gives an overview of the type of queries that are used. 
Results are depicted in Figure~\ref{fig:resultssgf}. Note that these queries all exhibit different properties.
Queries C1 and C2 both contain a set of SGF queries where
a number of atoms overlap. Query C3 is a complex query
that contains a multitude of different atoms. Finally, Query C4 
consists of two levels and many overlapping atoms.

We consider the following evaluation strategies in \Gumbo:
(\textit{i}) sequentially, i.e., one at a time, evaluating all BSGF
queries in a bottom-up fashion (\SEQUNIT); (\textit{ii}) evaluating
all BSGF queries in a bottom-up fashion level by level where queries
on the same level are executed in parallel (\PARUNIT); and, (\textit{iii})
using the greedily computed topological sort {combined with \gbsgf} (\gsgf);
\camera{Note that in \SEQUNIT and \PARUNIT all semi-joins are evaluated in separate jobs.}
For all tests conducted here, we found that \gsgf yields
multiway topological sorts that are identical to the
optimal topological sort (computed trough brute-force methods);
\camera{hence, we omit the results for the optimal plans.}


\begin{figure*}
\captionsetup[subfigure]{aboveskip=-5pt,belowskip=0pt}
\centering
\begin{subfigure}[b]{0.32\textwidth}
\begin{center}
\includegraphics[width=\linewidth]{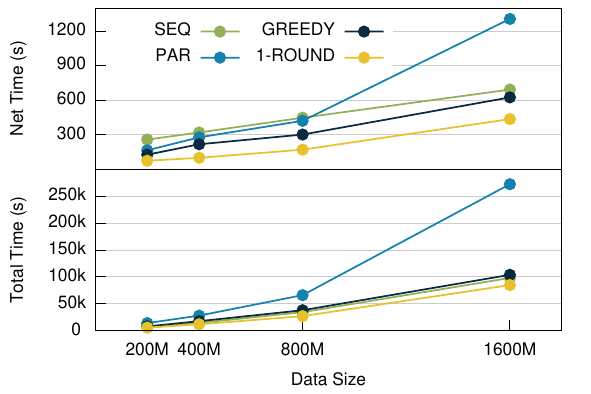}
\end{center}
 \caption{Varying data size (10 nodes).}
 \label{fig:sys:data}
\end{subfigure}
\begin{subfigure}[b]{0.32\textwidth}
\begin{center}
\includegraphics[width=\linewidth]{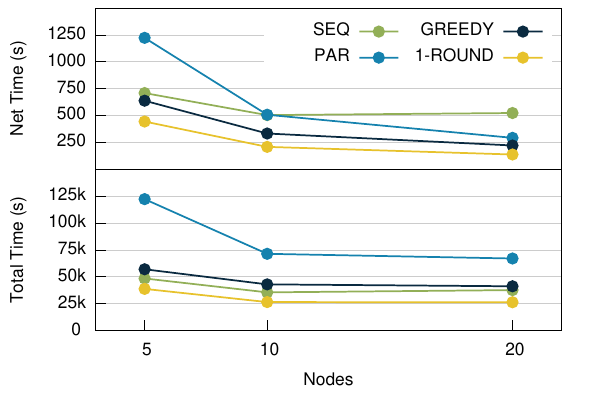}
\end{center}
 \caption{Varying cluster size (800M tuples).}
 \label{fig:sys:nodes}
\end{subfigure}
\begin{subfigure}[b]{0.32\textwidth}
\begin{center}
\includegraphics[width=\linewidth]{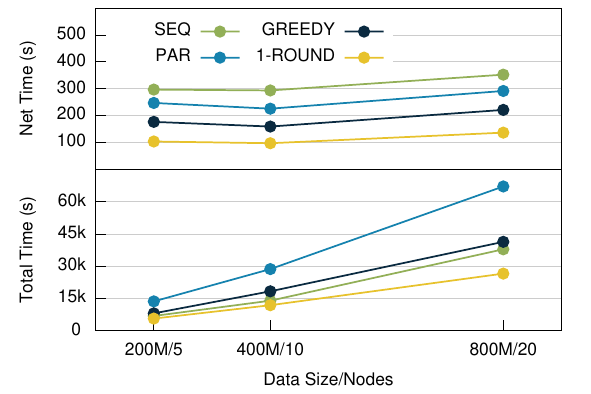}
\end{center}
 \caption{Varying data and cluster size.}
 \label{fig:sys:datanodes}
\end{subfigure}
\caption{Results for system characteristics tests for \Gumbo.}
\end{figure*}

Similar to our observations for BSGF queries,
we find that full sequential evaluation (\SEQUNIT)
results in the largest net times. Indeed, 
\PARUNIT exhibits 55\% lower net times on
average. We also observe  that \PARUNIT exhibits 
significantly larger total times than \SEQUNIT 
for queries C1 and C2, while this is not the case for C3 and C4.
The reason is that for C3 and C4, \camera{queries}
on the same level still share common characteristics, 
leading to a lower number of distinct semi-joins.

For \gsgf, we find that it exhibits net times
that are, on average, 42\% lower than \SEQUNIT,
while still being 29\% higher than \PARUNIT.
The main reason for this is the fact that \gsgf
aims to minimize total time, and may introduce
extra levels in the MR query plan to obtain this
goal. Indeed, we find that total times are down
27\% w.r.t.\ \SEQUNIT, and 29\% w.r.t.\ \PARUNIT.

Finally, we note that the absolute savings in net time
range from 115s to 737s for these queries,
far outweighing the overhead cost of calculating
the query plan itself, which typically takes around 10s {(sampling included)}.
Hence, we conclude that \gsgf provides \camera{an evaluation
strategy for SGF queries that manages to bring down the total time
(and hence, the resource cost) of parallel query plans,
while still exhibiting low net times when compared to 
sequential approaches.}

\subsection{System Characteristics} 
\label{sec:exp:system}
In this final experiment, we 
study the effect of \camera{growing} data size, cluster size, query size, and selectivity.
\camera{We choose queries similar to A3 to include the \ONEROUND strategy. 
Similar growth properties hold for the other query types.}

\paragraph*{Data \& Cluster Size}
Figures~\ref{fig:sys:data}-\ref{fig:sys:datanodes}
show the result of evaluating A3 using \SEQ, \PAR, \MSJ and \ONEROUND
under the presence of variable data and cluster size.
We summarize the most important observations:
\begin{compactenum}
\item in all scenarios, \ONEROUND performs best in terms
of net and total time;
\item due to its lack of grouping, \PAR needs a high number of mappers,
which at some point exceeds the capacity of the cluster. This causes a large
increase in net and total time, an effect that can be seen in Figure~\ref{fig:sys:data}.
\item with regard to net time, adding more nodes is very effective for the parallel strategies \PAR, \MSJ and \ONEROUND;  in contrast, adding more nodes does not improve \SEQ substantially after some point;
\item when scaling data and cluster size at the same time,
all strategies are able to maintain their net times in the presence
of an increasing total time.
\end{compactenum}

%
%

\begin{figure}
\begin{center}
\includegraphics[width=5.5cm]{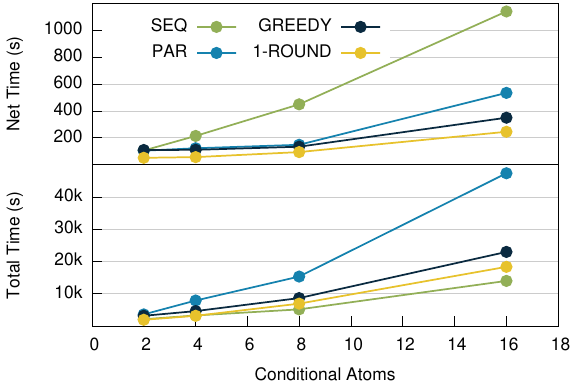}
\end{center}
 \caption{Varying the number of atoms.}
 \label{fig:sys:querysize}
\end{figure}

\paragraph*{Query Size} 
We consider a set of queries similar to A3, 
where the number of conditional atoms ranges from 2 to 16.
Results are depicted in Figure~\ref{fig:sys:querysize}.
With regard to net time, we find that
\SEQ shows an increase in net time 
that is strongly related to the query size, while
\PAR, \MSJ and \ONEROUND are less affected.
For total time,  we observe the converse for \PAR,
as this strategy cannot benefit from the packing 
optimization in the same way as \MSJ and \ONEROUND can.

\begin{table}[t]
\begin{center}
\small
\begin{tabular}{|r|rrr|rrr|}
\hline
&\multicolumn{3}{c|}{Net time}&\multicolumn{3}{c|}{Total time}\\
& A1&A2&A3& A1&A2&A3\\\hline
\SEQ 	&10\%	&9\%	&8\% 	&79\%	&95\%	&88\%\\
\PAR 	&33\%	&46\%	&69\%	&41\%	&47\%	&58\%\\
\MSJ 	&23\%	&30\%	&13\%	&45\%	&57\%	&15\%\\
\hline
\end{tabular}
\caption{Increase in net and total time when changing selectivity from 0.1 to 0.9 for queries A1--A3.}
\label{tab:selectivity}
\end{center}
\end{table}

\paragraph*{Selectivity} 
For a conditional relation, we define its 
\emph{selectivity rate} as the percentage of
guard tuples it matches. We tested queries
A1--A3 for selectivity rates 0.1 (high selectivity),
0.3, 0.5, 0.7 and 0.9 (low selectivity).
The increase in net time and total time between 
selectivity rates 0.1 and 0.9 is summarized 
in Table~\ref{tab:selectivity}.
In general, we find that the selectivity has the most influence
on the net times of \PAR and \MSJ,
and on the total times of \SEQ.
Finally, we observe that the filtering characteristics
of \SEQ disappear in the presence of low selectivity data,
causing total times to become comparable to \MSJ
for queries where packing is possible, such as A3.
This can be explained by \MSJ being less sensitive
to selectivity for queries where conditional atoms
share a common join key, making an effective
compression of intermediate data possible through
packing.

\section{Discussion}
\label{sec:concl}
We have shown that naive parallel evaluation of semi-join and (B)SGF queries can greatly reduce the net time of query execution, but, as expected, generally comes at a cost of an increased total time. We presented several methods that aim to reduce the total cost (total time) of parallel MR query plans, while at the same time avoiding a high increase in net time.
We proposed a two-tiered strategy for selecting the optimal parallel MR query plan for an SGF query in terms of total cost. As the general problem was proven to be NP-hard, we devised a two-tiered greedy approach that leverages on the existing technique of grouping MR jobs together based on a cost-model. The greedy approach was shown to be effective for evaluating (B)SGF queries in practice through several experiments using our own implementation called \Gumbo. {For certain classes of queries, our approach makes it even possible to evaluate (B)SGF queries in parallel with a total time similar to that of sequential evaluation.} We also showed that the profuse number of optimizations that are offered in \Gumbo allow it to outperform Pig and Hive in several aspects.

We remark that the techniques introduced in this paper generalize to any
map/reduce framework (as, e.g., \cite{XinRZFSS13}) given an appropriate adaptation of the cost model.

Even though the algorithms in this paper do not directly take skew into account, the presented framework can readily be adapted to do so when information on so-called heavy hitters is available or can 
be computed at the expense of an additional round~(see, e.g., \cite{GatesDN13,piglatin,hive,ramakrishnan2012}). 


 \paragraph*{Acknowledgment}
The third author was supported in part by grant no.\ NTU-ERP-105R89082D and the Ministry of Science and Technology Taiwan under grant no.\ 104-2218-E-002-038.
The computational resources and services used in this work were provided by the VSC (Flemish Supercomputer Center), funded by the Research Foundation - Flanders (FWO) and the Flemish Government -- department EWI. 
We thank Jan Van den Bussche \camera{and Jelle Hellings} for inspiring discussions and Geert Jan Bex for assistance  with cluster setup.

 \balance


\bibliographystyle{abbrv}
\bibliography{references-vldb2015} 

\begin{thebibliography}{10}

\bibitem{AHV95}
S.~Abiteboul, R.~Hull, and V.~Vianu.
\newblock {\em Foundations of Databases}.
\newblock Addison-Wesley, 1995.

\bibitem{AfratiJRSU-arxiv14}
F.~Afrati, M.~Joglekar, C.~R{\'{e}}, S.~Salihoglu, and J.~Ullman.
\newblock {GYM:} {A} multiround join algorithm in mapreduce.
\newblock {\em CoRR}, abs/1410.4156, 2014.

\bibitem{AfratiSSU13}
F.~Afrati, A.~D. Sarma, S.~Salihoglu, and J.~Ullman.
\newblock Upper and lower bounds on the cost of a map-reduce computation.
\newblock In {\em {VLDB}}, pages 277--288, 2013.

\bibitem{AfratiU11}
F.~Afrati and J.~Ullman.
\newblock Optimizing multiway joins in a map-reduce environment.
\newblock {\em {IEEE} Trans. Knowl. Data Eng.}, 23(9):1282--1298, 2011.

\bibitem{AfratiUV15}
F.~Afrati, J.~Ullman, and A.~Vasilakopoulos.
\newblock Handling skew in multiway joins in parallel processing.
\newblock {\em CoRR}, abs/1504.03247, 2015.

\bibitem{ABN98}
H.~Andreka, J.~van Benthem, and I.~Nemeti.
\newblock Modal languages and bounded fragments of predicate logic.
\newblock {\em Journal of Philosophical Logic}, 27(3):217--274, 1998.

\bibitem{BKS13}
P.~Beame, P.~Koutris, and D.~Suciu.
\newblock Communication steps for parallel query processing.
\newblock In {\em PODS}, pages 273--284, 2013.

\bibitem{BKS14}
P.~Beame, P.~Koutris, and D.~Suciu.
\newblock Skew in parallel query processing.
\newblock In {\em PODS}, pages 212--223, 2014.

\bibitem{BernsteinC81}
P.~Bernstein and D.~Chiu.
\newblock Using semi-joins to solve relational queries.
\newblock {\em J. {ACM}}, 28(1):25--40, 1981.

\bibitem{BernsteinG81}
P.~Bernstein and N.~Goodman.
\newblock Power of natural semijoins.
\newblock {\em {SIAM} Journal on Computing}, 10(4):751--771, 1981.

\bibitem{DBLP:journals/is/BernsteinG81}
P.~A. Bernstein and N.~Goodman.
\newblock The power of inequality semijoins.
\newblock {\em Inf. Syst.}, 6(4):255--265, 1981.

\bibitem{Blanas:2010bj}
S.~Blanas et~al.
\newblock {A comparison of join algorithms for log processing in MapReduce.}
\newblock {\em SIGMOD}, pages 975--986, 2010.

\bibitem{Chaudhuri98}
S.~Chaudhuri.
\newblock An overview of query optimization in relational systems.
\newblock In {\em {PODS}}, pages 34--43, 1998.

\bibitem{Chu:2015de}
S.~Chu, M.~Balazinska, and D.~Suciu.
\newblock From theory to practice: Efficient join query evaluation in a
  parallel database system.
\newblock In {\em SIGMOD}, pages 63--78, 2015.

\bibitem{DaenenNT15}
J.~Daenen, F.~Neven, and T.~Tan.
\newblock Gumbo: Guarded fragment queries over big data.
\newblock In {\em EDBT}, pages 521--524, 2015.

\bibitem{gumbo}
J.~Daenen and T.~Tan.
\newblock {Gumbo v0.4}, May 2016.
\newblock {\bf\url{http://dx.doi.org/10.5281/zenodo.51517}}.

\bibitem{Dean:2008:MSD:1327452.1327492}
J.~Dean and S.~Ghemawat.
\newblock {MapReduce}: Simplified data processing on large clusters.
\newblock {\em Commun. ACM}, 51(1):107--113, 2008.

\bibitem{Elseidy:2014ui}
M.~Elseidy, A.~Elguindy, A.~Vitorovic, and C.~Koch.
\newblock {Scalable and adaptive online joins.}
\newblock In {\em VLDB}, pages 441--452, 2014.

\bibitem{FlumFG02}
J.~Flum, M.~Frick, and M.~Grohe.
\newblock Query evaluation via tree-decompositions.
\newblock {\em J. {ACM}}, 49(6):716--752, 2002.

\bibitem{GareyJohnson}
M.~R. Garey and D.~S. Johnson.
\newblock {\em Computers and Intractability: A Guide to the Theory of
  NP-Completeness}.
\newblock W. H. Freeman \& Co., New York, NY, USA, 1979.

\bibitem{GatesDN13}
A.~Gates, J.~Dai, and T.~Nair.
\newblock Apache pig's optimizer.
\newblock {\em {IEEE} Data Engineering Bulletin}, 36(1):34--45, 2013.

\bibitem{Gradel98}
E.~Gr{\"a}del.
\newblock Description logics and guarded fragments of first order logic.
\newblock In {\em Description Logics}, 1998.

\bibitem{Ioannidis96}
Y.~Ioannidis.
\newblock Query optimization.
\newblock {\em {ACM} Computing Survey}, 28(1):121--123, 1996.

\bibitem{KoutrisS11}
P.~Koutris and D.~Suciu.
\newblock Parallel evaluation of conjunctive queries.
\newblock In {\em {PODS}}, 2011.

\bibitem{LMTV05}
D.~Leinders, M.~Marx, J.~Tyszkiewicz, and J.~{Van den Bussche}.
\newblock The semijoin algebra and the guarded fragment.
\newblock {\em Journal of Logic, Language and Information}, 14:331--343, 2005.

\bibitem{nykiel2010}
T.~Nykiel, M.~Potamias, C.~Mishra, G.~Kollios, and N.~Koudas.
\newblock {MRShare}: Sharing across multiple queries in mapreduce.
\newblock In {\em VLDB}, pages 494--505, 2010.

\bibitem{Okcan:2011kb}
A.~Okcan and M.~Riedewald.
\newblock {Processing theta-joins using MapReduce.}
\newblock In {\em SIGMOD}, pages 949--960, 2011.

\bibitem{OlstonRSS08}
C.~Olston, B.~Reed, A.~Silberstein, and U.~Srivastava.
\newblock Automatic optimization of parallel dataflow programs.
\newblock In {\em {USENIX} Annual Technical Conference}, pages 267--273, 2008.

\bibitem{piglatin}
C.~Olston, B.~Reed, U.~Srivastava, R.~Kumar, and A.~Tomkins.
\newblock Pig latin: a not-so-foreign language for data processing.
\newblock In {\em SIGMOD}, pages 1099--1110, 2008.

\bibitem{PicalausaFHV14}
F.~Picalausa, G.~Fletcher, J.~Hidders, and S.~Vansummeren.
\newblock Principles of guarded structural indexing.
\newblock In {\em {ICDT}}, pages 245--256, 2014.

\bibitem{ramakrishnan2012}
S.~R. Ramakrishnan, G.~Swart, and A.~Urmanov.
\newblock Balancing reducer skew in mapreduce workloads using progressive
  sampling.
\newblock In {\em SoCC}, page~16. {ACM}, 2012.

\bibitem{TLX13}
Y.~Tao, W.~Lin, and X.~Xiao.
\newblock Minimal mapreduce algorithms.
\newblock In {\em SIGMOD}, pages 529--540, 2013.

\bibitem{hive}
A.~Thusoo, J.~S. Sarma, N.~Jain, Z.~Shao, P.~Chakka, N.~Zhang, S.~Anthony,
  H.~Liu, and R.~Murthy.
\newblock Hive - a petabyte scale data warehouse using hadoop.
\newblock In {\em ICDE}, pages 996--1005, 2010.

\bibitem{Vardi96}
M.~Vardi.
\newblock Why is modal logic so robustly decidable?
\newblock In {\em DIMACS Workshop on Descriptive Complexity and Finite Models},
  pages 149--184, 1996.

\bibitem{wang2013}
G.~Wang and C.-Y. Chan.
\newblock Multi-query optimization in mapreduce framework.
\newblock In {\em VLDB}, pages 145--156, 2013.

\bibitem{White15}
T.~White.
\newblock {\em Hadoop - The Definitive Guide: Storage and Analysis at Internet
  Scale}.
\newblock O'Reilly, 2015.

\bibitem{XinRZFSS13}
R.~Xin, J.~Rosen, M.~Zaharia, M.~Franklin, S.~Shenker, and I.~Stoica.
\newblock Shark: {SQL} and rich analytics at scale.
\newblock In {\em {SIGMOD}}, pages 13--24, 2013.

\bibitem{Yannakakis81}
M.~Yannakakis.
\newblock Algorithms for acyclic database schemes.
\newblock In {\em VLDB}, pages 82--94, 1981.

\end{thebibliography}


\newpage
\onecolumn
\appendix

\section{Proof of Theorem~2}
\label{appA}
In order to prove Theorem~\ref{theorem2}, we consider a more general problem. 
The same technique can then be used to prove the original theorem.
We define the optimization problem \sgfgopt that takes as input a set $S$
and a cost function $w: 2^{|S|} \rightarrow \mathbb{N}$,
and constructs a partition $\cS = \{S_1, \ldots, S_n\}$ of $S$ 
such that $\sum_{S_i\in \cS} w(S_i)$ is minimized.
The decision version of \sgfgopt, denoted by \sgfg, 
corresponds to deciding, for a given positive integer $k$,
whether there exists a partition $\cS = \{S_1, \ldots, S_n\}$ 
of $S$ such that $\sum_{i}^n w(S_i) = k$. 
We now prove the NP-completeness of \sgfg.

\begin{theorem}
\label{sgfgnpc}
The \sgfg problem is NP-complete.
\end{theorem}

\begin{proof}
The problem is clearly in NP, as a given solution
can be verified in polynomial time.
We prove that this problem is NP-complete by 
a providing a polynomial time reduction from the {\sc Subset Sum} problem, 
which consists of deciding, for given an integer $k$ and a set of positive integers 
$A$, whether there exists a set $B \subseteq A$ such that
$\sum_{b\in B} b = k$. This problem is NP-complete~\cite{GareyJohnson}.

The reduction is as follows. Given an instance of the subset sum problem,
i.e., a set of integers $A$ and an integer $k$, we construct the following instance 
of \sgfg:
The set of items $S$ equals $A \cup \{\specialitem\}$, 
where $\specialitem$ is an element not present in $A$,
and the cost function $w:2^{S}\to\mathbb{S}$ is defined as
\begin{equation}
\label{eq:wdef}
    w(X) = 
	\begin{cases}
               \gamma			& \mbox{if}\ \specialitem \in X\\
               \sum_{a \in X} a	& \mbox{otherwise}\\

	\end{cases}
\end{equation}
Here, $\gamma$ is an arbitrary constant. Intuitively,
$w$ takes the sum of all items in a given set, except when
the set contains the special item $\specialitem$. In that case, the value
will be a fixed number $\gamma$.

In the remainder of this proof, we show that there exists 
a $B \subseteq A$ such that $\sum_{b\in B} b = k$
iff  there exists a partition $\cS = \{S_1, \ldots, S_n\}$ 
of $S$ such that $\sum_{i}^n w(S_i) = k + \gamma$. 

Suppose there exists a partition $\cS = \{S_1, \ldots, S_n\}$ 
of $S$ such that $\sum_{i}^n w(S_i) = k + \gamma$.
Now, select the element of the partition that 
contains $\specialitem$, say $S_{\specialitem}$, and remove it from the partition
to obtain a partition $\cS'$ for $S \setminus S_{\specialitem}$. 
Now, let $B = \bigcup \cS'$ and note that $B \subseteq A$.
Clearly, the sum of the elements in $B$ equals $k + \gamma - \gamma = k$,
as the cost of the removed element equals $\gamma$ by Equation~\eqref{eq:wdef}.
Hence, $B$ is a $k$-cost subset of $A$.

Suppose there exists  a $B \subseteq A$ 
such that $\sum_{b\in B} b = k$.
Consider the partition $\cS = \{B, C\}$ of $S$,
where $C = (A \setminus B) \cup \{\specialitem\}$.
As $w(B) = k$ and $w(C) = \gamma$, 
we have a total cost of $k+\gamma$.
\end{proof}

We conclude this section by providing a proof for Theorem~\ref{theorem2}.
Consider the decision version of \sgfopt, which we'll denote by \sgf:
for a DAG of BSGF queries $\cG_Q$ and a positive integer $k$
determine whether there exists a multiway topological sort $\cF$ of $\cG_Q$ 
with $\cost(\cF) = k$.

\begin{theorem}
The \sgf problem is NP-complete.
\end{theorem}
\begin{proof}
We can verify a solution for this problem in the following way:
given a multiway topological sort $\cF$ and a positive integer $k$, we can calculate its total cost using
Equation~\eqref{eq:cost-sgf-gopt}, i.e.,
\[
\cost(\cF) \ = \
\sum_{i=1}^k \ \cost(\gopt(F_i)),
\]
and check whether it equals $k$.
As the greedy algorithm \gbsgf for $\gopt$ runs in polynomial time, the total cost can be calculated in
polynomial time as well and hence the problem is in NP.

We now provide a polynomial time reduction from {\sc Subset Sum} to \sgf.
Given an instance of the {\sc Subset Sum} problem consisting of a set of positive integers
$A = \{a_1,\ldots, a_n\}$ and a positive integer $k$, we construct the following instance for \sgf.
Let $R_1, \ldots, R_n, R^{\specialitem}$ be a set of binary relations containing no tuples, a set of relations $S_1, \ldots, S_n$ be a set of binary relations where $|S_i| = a_i$, each tuple has a size of 1MB and every second field has value 0. 
Also, consider a set of BSGF queries 
$F = \{f_1, \ldots, f_n, f^{\specialitem}\}$, where $f_i$ equals $R_i(x_i,y_i) \semijoin S_i(x_i,1)$
and
\[
f^{\specialitem} = R^{\specialitem}(x,1) \semijoin R_1(x_1,y_1) \land \ldots \land R_n(x_n,y_n) \land S_1(x_1,1) \land \ldots \land S_n(x_1,1).
\]
Notice that there are no dependencies between the queries.
Finally, for the Hadoop system, we choose all I/O constants equal to 0, except $\Hrc=1$.

Using Equation~\eqref{eq:cost-bsgf} and our cost model, we obtain that the cost of calculating each individual BSGF query $f_i$ equals $a_i$:
\begin{align*}
\cost(\gopt(\{f_i\}))
		&= \cost(\eval(R_i,\varphi_{C_i})) + \cost(\{f_i\})\\
		&= 0 + \cost_{\map}(|f_i|,|f_i|) + \cost_{\red}(|f_i|,|f_i|)\\
		&= 0 + \Hrc \cdot a_i + 0\\
		&= a_i.
\end{align*}
Note that the \eval job has cost 0, as no output is provided by the \validate job.

We also find that the cost of two jobs equals their individual sums, regardless whether or not they are grouped in the query plan provided by \gopt:
\begin{align*}
\cost(\gopt(\{f_i,f_j\})) &= a_i + a_j.
\end{align*}

Finally, we note that \gopt will always group $f_i$ with $f^{\specialitem}$ as all relations of $f_i$ appear in $f^{\specialitem}$. This leads to a cost of:
\begin{align*}
\cost(\gopt(\{f_i,f^{\specialitem}\})) &= \gamma,
\end{align*}
where $\gamma = \sum_{i=1}^n a_i$.

Let $\cG_F$ be the dependency graph that has nodes $F$ and no edges (as there are no query dependencies).
Now, there exists a multiway topological sort for $\cG_F$ of cost $k + \gamma$ iff there exists
a $B \subseteq A$ such that $\sum_{b\in B} b = k$.

Observe that the cost function behaves completely analogous to the one used in the proof of Theorem~\ref{sgfgnpc}.
Hence, the same reasoning can now be applied to complete the proof.
\end{proof}

%

\section{Hadoop Settings \& Cost Model Constants}
\label{appB}
Table~\ref{tab:hadoopsettings} lists the settings for the Hadoop cluster that was utilized for our experiments.
Table~\ref{tab:realcostconstants} shows the values that are used for the constants in the revised MapReduce cost model. The values were obtained using benchmark tests on system used for the experiments.

\begin{table}[h]
\begin{center}
    \begin{tabular}{|l|l|}
    \hline
    Hadoop Setting & Value \\\hline\hline
    \texttt{io.file.buffer.size}                          & 131072   \\\hline
    \texttt{dfs.replication}                              & 3          \\\hline
    \texttt{mapred.child.java.opts}                       & -Xmx1024m \\\hline
    \texttt{mapreduce.map.memory.mb}                      & 1280       \\\hline
    \texttt{mapreduce.reduce.memory.mb}                   & 1280      \\\hline
    \texttt{mapreduce.task.io.sort.mb}                    & 512       \\\hline
    \texttt{mapreduce.reduce.merge.inmem.threshold}       & 0          \\\hline
    \texttt{mapreduce.reduce.input.buffer.percent}        & 0.5        \\\hline
    \texttt{mapreduce.job.reduce.slowstart.completedmaps} & 1          \\\hline
    \texttt{mapred.map.tasks.speculative.execution}         & false         \\\hline
    \texttt{mapred.reduce.tasks.speculative.execution}         & false         \\\hline
    \texttt{yarn.nodemanager.resource.memory-mb}          & 49152      \\\hline
    \texttt{yarn.scheduler.minimum-allocation-mb}         & 4096       \\\hline
    \texttt{yarn.scheduler.maximum-allocation-mb}         & 49152      \\\hline
    \texttt{yarn.nodemanager.resource.cpu-vcores}         & 10         \\\hline
    \end{tabular}
    \end{center}
    \caption{Hadoop settings used in the experiments.}
\label{tab:hadoopsettings}
\end{table}

\begin{table}[h]
\centering
\begin{tabular}{|c|l|r|}
\hline
Constant &Description&Value\\\hline\hline
$\Lrc$    & local disk read cost (per MB)&0.03 \\\hline	
$\Lwc$    & local disk write cost (per MB)&0.085 \\\hline
$\Hrc$    & hdfs read cost (per MB)&0.15\\\hline		
$\Hwc$    & hdfs write cost (per MB)&0.25\\\hline	
$\Trc$    & transfer cost  (per MB)&0.017\\\hline	\hline
$D$ & external sort merge factor&10\\\hline
$\mbuflim$ & map task buffer limit (in MB)& 409MB\\\hline
$\rbuflim$ & reduce task buffer limit (in MB)& 512MB\\\hline	
\end{tabular}
\caption{Description of constants used in the cost model.
}
\label{tab:realcostconstants}
\end{table}

\end{document}